\documentclass{article}
\usepackage{amsfonts}
\usepackage{amssymb}
\usepackage{amsmath,mathrsfs}

\usepackage{amsthm}
\usepackage{latexsym}
\usepackage{layout}
\usepackage{indentfirst}
\usepackage{fancyhdr}
\usepackage{epsfig}
\usepackage{graphicx}
\usepackage{pdfsync}
\usepackage{cite}
\usepackage{amssymb}
\usepackage{dsfont}
\usepackage{caption}
\usepackage{subcaption}
\usepackage{color}

\usepackage[english]{babel}

\theoremstyle{plain}
\begingroup
\newtheorem{theorem}{Theorem}[section]
\newtheorem{lemma}[theorem]{Lemma}
\newtheorem{proposition}[theorem]{Proposition}

\endgroup

\theoremstyle{definition}
\begingroup

\newtheorem{remark}[theorem]{Remark}
\newtheorem{example}[theorem]{Example}
\endgroup

\theoremstyle{remark}

\mathchardef\emptyset="001F

\numberwithin{equation}{section}

\newcommand{\dt}{\,\mathrm{d}t}

\newcommand{\e}{\varepsilon}
\newcommand{\Ln}{\lambda_n}
\newcommand{\dn}{\delta_n}
\newcommand{\R}{\mathbb{R}}
\newcommand{\rthree}{\R^3}
\newcommand{\Z}{{\mathbb{Z}}}
\newcommand{\tu}{\tilde{u}}
\newcommand{\N}{{\mathbb{N}}}

\newcommand \q{\quad}

\def\endproof{\qed}

\title{Chirality transitions in frustrated $S^{2}$-valued spin systems}

\author{Marco Cicalese\thanks{Zentrum Mathematik - M7, Technische Universit\"at M\"unchen, Boltzmannstrasse 3, 85748 Garching, Germany. Email: {\tt cicalese@ma.tum.de}} \and Matthias Ruf\thanks{Zentrum Mathematik - M7, Technische Universit\"at M\"unchen, Boltzmannstrasse 3, 85748 Garching, Germany. Email: {\tt mruf@ma.tum.de}}\and Francesco Solombrino\thanks{Zentrum Mathematik - M7, Technische Universit\"at M\"unchen, Boltzmannstrasse 3, 85748 Garching, Germany. Email: {\tt francesco.solombrino@ma.tum.de}}}

\begin{document}

\maketitle

\begin{abstract}
We study the discrete-to-continuum limit of the helical XY $S^{2}$-spin system on the lattice $\Z^{2}$. We scale the interaction parameters in order to reduce the model to a spin chain in the vicinity of the Landau-Lifschitz point and we prove that at the same energy scaling under which the $S^{1}$-model presents scalar chirality transitions, the cost of every vectorial chirality transition is now zero. In addition we show that if the energy of the system is modified penalizing the distance of the $S^{2}$ field from a finite number of copies of $S^{1}$, it is still possible to prove the emergence of nontrivial (possibly trace dependent) chirality transitions. 
\end{abstract}

\section{Introduction}
In the last decades frustrated spin systems with continuous symmetry have attracted a great interest both in the physical and in the mathematical community as simple models leading to helical phases, which turn out to be interesting for possible application as multi-ferroics (see \cite{CM} for a recent review on the subject). Despite a great effort, the phase diagram of these systems is far from being rigorously described. In this paper we consider the helical XY spin model (see \cite{V}) on the square lattice $\mathbb{Z}^2$ as a prototype of such systems and we scale the interaction parameters in order to study, by variational techniques, the vicinity of the Landau-Lifschitz point, where the helical behavior is expected, as the continuum limit is approached.    

\bigskip

A configuration for the helical XY spin model on the square lattice $\mathbb{Z}^2$ is a map $u:i\in\mathbb{Z}^2\mapsto u^{i}\in S^2$ whose energy reads 
\begin{equation}\label{HXY}
E(u)=-\sum_{i\in\Z^{2}}J_0 (u^{i},u^{i+e_1})-J_1(u^{i},u^{i+2e_1})+J_2(u^{i},u^{i+e_2}),
\end{equation}
where $J_0$ and $J_{1}$ are the interaction parameters for the nearest-neighbors (NN) and the next-to-nearest-neighbors (NNN) interactions in the direction horizontal $e_{1}$, respectively, while $J_{2}$ is the interaction parameter for the NN interactions in the vertical direction $e_{2}$.  Note that the behavior of the functional above strongly depends on the values of the interaction parameters and on the range of the spin field. For instance in the case $J_{1}=0, J_{0}=J_{2}>0$ and $S^{1}$-valued spins, one recovers the classical XY-model whose discrete-to-continuum limit has been investigated in the variational framework of $\Gamma$-convergence in \cite{ACXY} (see also \cite{ACP} and \cite{AP}).

In the present paper we consider $J_0, J_{1}, J_{2}>0$. With this choice the behavior of the system in the two directions is different. In the direction $e_{2}$ the system is ferromagnetic, the interaction potential is $-J_2(u^{i},u^{i+e_2})$ and favors spin alignment. In the direction $e_{1}$ there are competing ferromagnetic (F) NN interactions whose potential $-J_0 (u^{i},u^{i+e_1})$ favors alignment and anti-ferromagnetic (AF) NNN interactions with potential $J_1(u^{i},u^{i+2e_1})$ favoring antipodal spins. This competition acts as a source of frustration. More precisely, along each horizontal line, the energy accounting for interactions in the $e_{1}$-direction, namely 
\begin{equation}\label{intro:energy-1d}
F(u)=-\sum_{j\in\Z}J_0 (u^{j},u^{j+1})-J_1(u^{j},u^{j+2}),
\end{equation}
is that of a so called F/AF {\it frustrated chain} (note that for $J_{2}=0$ the system would behave as a collection of independent chains). In this context we say that $F$ is frustrated because there isn't any configuration minimizing all the interactions at once (see \cite{diep} for a comprehensive study of frustrated spin systems).

In this paper we study the functional \eqref{HXY} under the natural scaling of the interaction parameters leading to the easiest possible geometry for helical ground states (for other possible scaling in a continuous approximation see \cite{SPN}). To this end we enforce alignment of the spins in the direction $e_{2}$ by letting $J_2$ diverge positively. As a result, finite energy spin fields $u$ have a one-dimensional profile; i.e., $u^{(i_{1},i_{2})}=v^{i_{1}}$ for some $v:\Z\to S^{2}$.  In other words, the system can be, modulo technicalities, described by studying the behavior of the one-dimensional F/AF frustrated chain model for $S^{2}$-valued spins. For the latter chain model it has been conjectured in the appendix of \cite{DmiKriext} (the extended version of \cite{DmiKri}) that, when $J_0$ and $J_1$ are close to the helimagnetic transition point $J_{0}/J_{1}=4$, the system presents chirality transitions as in the case of $S^{1}$-valued spins whose variational analysis has been recently carried out by 
the first and the third authors in \cite{ciso}. In the present paper we disprove this conjecture showing that in the $S^{2}$ case the transition energy between ground states with different chiralities is negligible. Furthermore we propose an alternative minimal model leading to non trivial chirality transitions.

\medskip

In \cite{ciso} the continuum limit of the F/AF chain energy in \eqref{intro:energy-1d} has been studied in the case of $S^{1}$-valued spins and for a range of interaction parameters close to the ferromagnetic/helimagnetic transition point. The outcome of the analysis is summarized below. After scaling the functional by a small parameter $\frac{\Ln}{J_{1}}$ ($\Ln\to0$ as $n\to\infty$), and setting $\Z_{n}=\{j \in \Z:\  \lambda_{n} j \in [0,1]\}$ one defines $F_{n}:\{u:j\in\Z_{n}\mapsto u^{j}\in S^{1}\}\to\R$ as 
\begin{equation}\label{intro:energy-n}
F_{n}(u)=-\alpha\sum_{j\in\Z_{n}}\Ln(u^{j},u^{j+1})+\sum_{j\in\Z_{n}}\Ln(u^{j},u^{j+2}).
\end{equation}
where $\alpha=J_{0}/J_{1}$ is the so called frustration parameter. 
It turns out that the ground states of $F_{n}$ can be completely characterized. Neighboring spins are aligned if $\alpha\geq 4$ (ferromagnetic order), while they form a constant angle $\varphi=\pm\arccos(\alpha/4)$ if $0<\alpha<4$ (helimagnetic order). In this last case the system shows a chirality symmetry: the two possible choices of $\varphi$ correspond to either clockwise or counter-clockwise spin rotations, or in other words to a positive or a negative chirality. The energy necessary to break this symmetry as $\alpha$ is close to $4$ can be found letting the frustration parameter $\alpha$ depend on $n$ and replacing in \eqref{intro:energy-n} $\alpha$ by $\alpha_{n}=4(1-\dn)$ for some vanishing sequence $\dn>0$. One then introduces the renormalized energy 
\begin{equation}\label{intro:Hndelta}
H_{n}(u)=\frac{1}{2}\sum_{j\in\Z_{n}}\Ln\left|u^{j+2}-2(1-\dn)u^{j+1}+u^{j}\right|^{2},
\end{equation}
proves that, under periodic boundary conditions on the scalar product of NN interactions, 
\begin{equation}\label{intro:Ehf-to-Hhf}
H_{n}(u)=F_{n}(u)-\min F_{n}
%=F_{n}(u)+(3-4\dn+2\dn^{2})(1-c_{n}\Ln)
\end{equation}
and computes the $\Gamma$-limit of $H_{n}/(\Ln\dn^{3/2})$ with respect to the $L^{1}$ convergence of the chirality order parameter (a proper discrete version of the angular increment between two neighboring spins) as $\Ln\to 0$. In the case $\Ln/\sqrt{\dn}\to 0$ (at other scalings chirality transitions are either forbidden or not penalized) the limit energy functional is proportional to the number of jumps of the chirality, namely the number of times the spin configuration changes the sign of its angular velocity.
%Because of the different frustration in the two directions this model is said to present frustration anisotropy. 
\medskip

In the case of $S^{2}$-valued spins the picture drastically changes. In analogy with the $S^{1}$ case described above one may still renormalize the energy and prove that the modulus of (a proper discrete version of) the angular velocity of a ground state, which one may still interprets as the chirality of the system, is constant. However, we may now prove that at this scaling the transition energy between two ground states with different chiralities is zero. The proof uses the fact that, in contrast to the $S^{1}$-case, the $S^{2}$-spin system does not need to jump from one chiral state to another in order to modify its chirality. Instead, it lets the chirality vary on a slow scale paying little energy (see Figure \ref{Fig.1}). This is proved in Theorem \ref{main1} exploiting the fact that, at leading order, the renormalized energy can be rewritten as a discrete vectorial Modica-Mortola functional presenting a potential term with connected wells. Note that in the continuous setting, the analysis by $\Gamma$-
convergence of such functionals has been performed by \cite{Ambrosio} and \cite{Baldo}. However the discreteness of our energies as well as the additional differential constraint defining our order parameter prevents us from directly using the results contained in those papers. 

\begin{figure}
\includegraphics[width=0.8\textwidth]{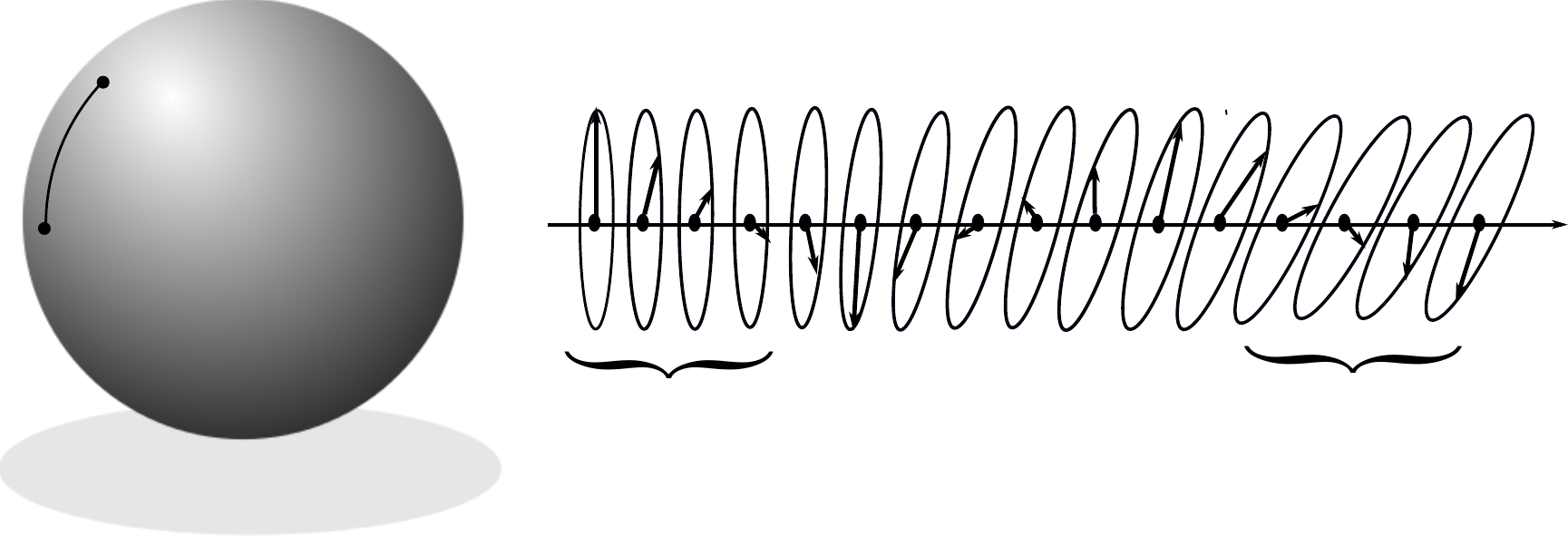}
\caption{Connecting the chiralities $q_{1}$ and $q_{2}$ by slowly moving the spin rotation axis}\label{Fig.1}
\begin{picture}(0,0)
\put(-10,100){$q_{1}$}\put(3,135){$q_{2}$}
\put(148,61){$q_{1}$}\put(308,61){$q_{2}$}
\end{picture}
\end{figure}
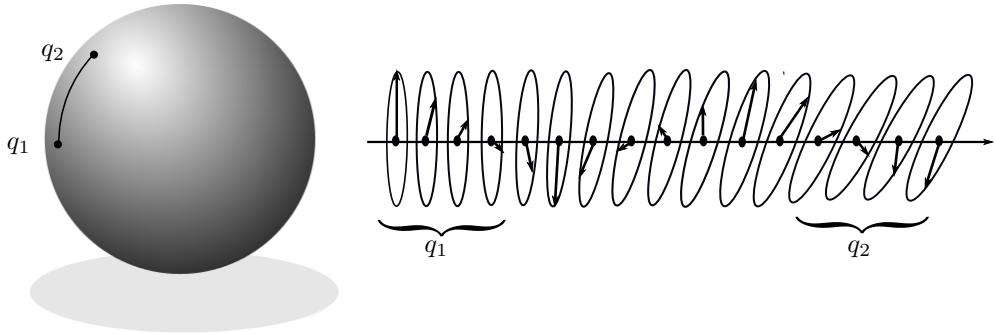

In the second part of the paper we propose and study two possible spin models leading to nontrivial chirality transitions in the vicinity of the ferromagnetic/helimagnetic transition point. To this end we modify the functional $H_{n}$ by adding what we call either a {\it hard} or a {\it soft} penalization term. In the {\it hard} case we constrain the spin variable to take values only in a subset of $S^{2}$ consisting of finitely many copies of $S^1$, while in the {\it soft} case we penalize the distance of the spin field from such a set. For the first model we show that the optimal transition is obtained by first slowing down the angular velocity of the spin field in the first phase until it reaches one intersection point between the two rotation planes between which the transition occurs with zero velocity, and then speeding up again the angular velocity in the new phase (see Figure \ref{Fig.2}). In terms of chirality, the transition corresponds to first decreasing and then increasing the length of the 
chirality vector while keeping its orientation constant in each phase. For the second model the construction is more involved and the optimal path may, depending on the scaling of the additional penalization term, be either again the one described in Figure \ref{Fig.2} or instead depend on the shape of the penalization potential. As a result, the limit functionals obtained with the two proposed models are different: while in the first case the chirality transitions lead to a constant positive limit energy to be paid for each discontinuity in the chirality (no matter which chiralities the system is trying to connect), in the second one, under appropriate scaling, the limit energy may depend on the two transition chiral states (see example \ref{anglenonconstant}). As a final technical remark, we notice that the analysis of the discrete-to-continuum limit for the second model can be seen as a generalization in the vector-valued case of some results concerning the discrete approximation of Modica-Mortola type 
functionals obtained in \cite{BraYip}.

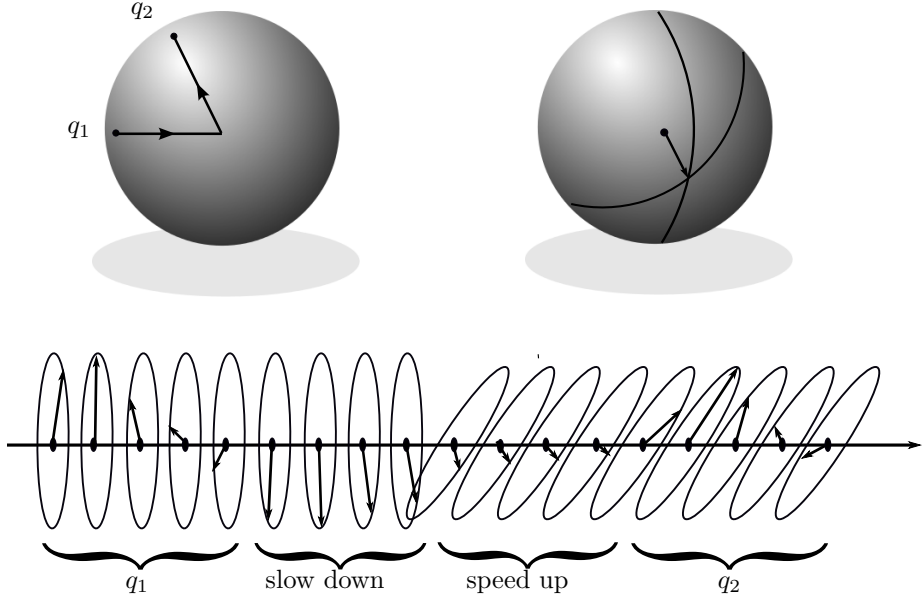
\begin{figure}
\includegraphics[width=0.8\textwidth]{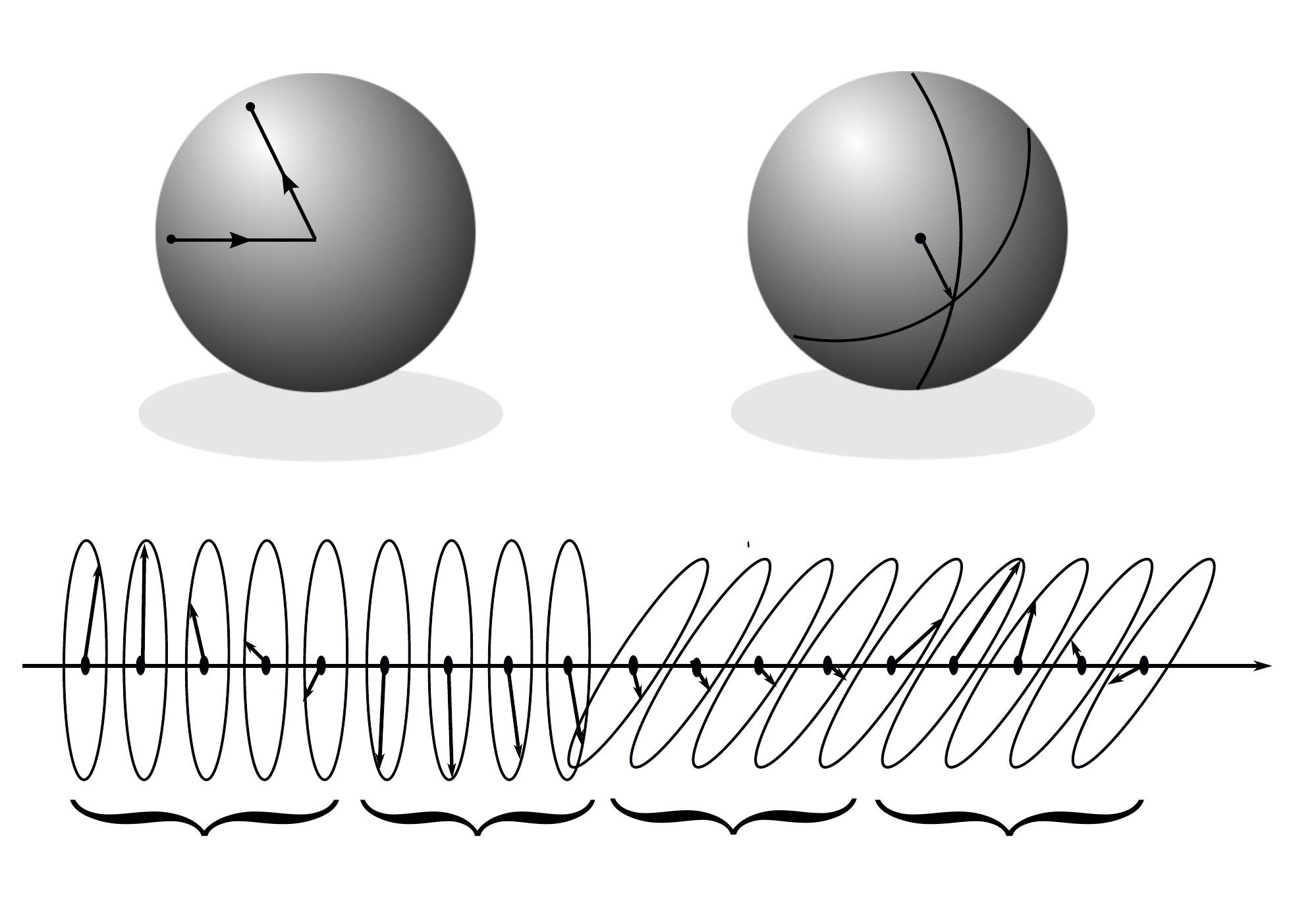}
\caption{Chirality transition between $q_{1}$ and $q_{2}$. Ball on the left: the transition path in the chirality space. Ball on the right: the intersection of the two $S^{1}$ on the {\em transition direction}. Bottom: the transition in the real space. The chirality slows down in the starting rotation plane until the spin reaches the transition direction and then it speeds up again in the final rotation plane.} \label{Fig.2}
\begin{picture}(0,0)
\put(52,295){$q_{2}$}\put(28,250){$q_{1}$}
\put(50,78){$q_{1}$}\put(274,78){$q_{2}$}
\put(104,78){slow down}\put(180,78){speed up}
\end{picture}
\end{figure}

\section{The energy model: preliminary considerations}

\subsection{Basic notation}

Let $\Omega=(0,1)^2\subset\R^2$ and $\Ln$ a vanishing sequence of positive numbers. We set $\mathbb{Z}^2_n(\Omega)$ as the set of those $i\in\Z^{2}$ such that $\Ln i\in \Ln\mathbb{Z}^2\cap\Omega$ and $R_n(\Omega):=\{i\in \Z^2_n(\Omega):\;i+2e_1,i+e_2\in\Z^{2}_{n}(\Omega\}$. The symbol $B(0,1)$ stands for the unitary ball of $\R^{3}$ centered at the origin. The symbols $S^1,\,S^{2}$ stand as usual for the unit spheres of $\R^2$ and $\R^3$, respectively. Given two vectors $a,b\in\R^{3}$ we will denote by $(a,b)$ their scalar product. Moreover we define $\mathcal{U}^2_n(\Omega)$ as the space of functions $u:i\in\Z_n^2(\Omega)\mapsto u^i\in S^2$ and $\overline{{\mathcal{U}}}^2_n(\Omega)$ as the subspace of those functions $u$ such that, for all $m\in \Ln\mathbb{Z}\cap(0,1)$, it holds
\begin{equation}\label{slicedboundary}
(u^{(i_{min}+1,m)},u^{(i_{min},m)})=(u^{(i_{max},m)},u^{(i_{max}-1,m)}),
\end{equation}
where $i_{min}$ and $i_{max}$ are the minimum and the maximum of $\Ln\Z\cap(0,1)$, respectively. 

\subsection{The energy}

As pointed out in the introduction, we let the parameters $J_0,\,J_1,\,J_2$ in (\ref{HXY}) be scale dependent. Without loss of generality we divide the energy by $J_{1,n}>0$ and rename $J_{0,n}$ and $J_{2,n}$, accordingly. Given $n\in\mathbb{N}$ and a function $u\in\overline{{\mathcal{U}}}^2_n(\Omega)$ we consider the energy

\begin{equation}\label{HXYn}
E_n(u)=-\sum_{i\in R_n(\Omega)}\Ln^2\big(J_{0,n} (u_i,u_{i+e_1})-(u_i,u_{i+2e_1})+J_{2,n}(u_i,u_{i+e_2})\big).
\end{equation}
We remark that, by considering the energy defined on $u\in\overline{{\mathcal{U}}}^2_n(\Omega)$, we are imposing boundary conditions only in the $e_{1}$-direction, while we are leaving the spin field in the $e_{2}$-direction unconstrained. As a matter of fact, as a result of the scaling we are going to choose, constraining the spins in the $e_{2}$-direction would not affect the asymptotic energy.

\subsection{Ground states and renormalized energy}
In this paragraph we describe the ground states of the energy $E_{n}$ and compute their energy $\min E_{n}$. We then define a {\it renormalized} energy $H_{n}$ which will be the main object to study in order to discuss the asymptotic behavior of the system in the next sections.

We begin observing that the minimizers of $E_n$ can be easily computed if one knows the minimizers of the energy accounting for the interactions in the $e_{1}$-direction only. Indeed the ground states of the system are then obtained by extending such minimizers constantly in the $e_{2}$-direction as it is explained below. Note that indeed this extension keeps the ferromagnetic term in the $e_{2}$-direction minimal.

We now find the ground states of $E_{n}$ adapting the idea in the proof of Proposition $3.2$ of \cite{ciso}. We repeat the argument for the reader's convenience.   Setting the renormalized energy $H_{n}$ as 
\begin{equation*}
H_n(u)=\frac{1}{2}\left(\sum_{i\in R_n(\Omega)}\Ln^2|u^i-\frac{J_{0,n}}{2}u^{i+e_1}+u^{i+2e_1}|^2+J_{2,n}\sum_{i\in R_n(\Omega)}\Ln^2|u^{i+e_2}-u^i|^2\right),
\end{equation*}
we observe that 
\begin{equation}
E_{n}(u)=H_n(u)-\left(1+\frac{J^{2}_{0,n}}{8}+J_{2,n}\right)(1-a_{n}),
\end{equation}
where $a_{n}=1-\sum_{i\in R_{n(\Omega)}}\Ln^{2}$ is such that $\lim_{n}a_{n}=0$.\\
In the case $J_{0,n}\leq 4$ we take $\phi_{n}\in\left[-\frac{\pi}{2},\frac{\pi}{2}\right]$ such that $\cos\phi_{n}=(J_{0,n}/4)$ and, for all $i=(l,m)$, we define the three dimensional vector $u_{n}^{i}$ as 
\begin{equation}\label{ground}
u_{n}^{i}=(\cos(\phi_{n}l)|\sin(\phi_{n}l)|0).
\end{equation}
We have that $u_{n}^{i+e_{2}}-u_{n}^{i}=0$ for all $i$ while, by means of trigonometrical identities it holds that 
\begin{equation}
u_{n}^i-\frac{J_{0,n}}{2}u_{n}^{i+e_1}+u_{n}^{i+2e_1}=0,
\end{equation}
therefore $u_{n}$ is a ground state and $\min E_{n}=E_{n}(u_{n})=-\left(1+\frac{J^{2}_{0,n}}{8}+J_{2,n}\right)(1-a_{n})$. By the rotational invariance of the energy, all those states obtained rotating $u_{n}$ by a fixed $SO(3)$ matrix are ground states, too. \\
Conversely, let $v_{n}$ be a ground state of $E_{n}$. We have that $H_{n}(v_{n})=0$ which implies that 
\begin{eqnarray}
v_{n}^{i+e_1}&=&\frac{2}{J_{0,n}}(v_{n}^i+v_{n}^{i+2e_1}),\label{equality_{1}}\\
v_{n}^{i+e_{2}}&=&v_{n}^{i}, \nonumber
\end{eqnarray}
so that $v_{n}^{i}$ is independent on the vertical coordinate and lies on a fixed plane. 
By taking the modulus squared in \eqref{equality_{1}} we further get that 
$$
1=\frac{4}{J_{0,n}^{2}}|v^{i}_n+v^{i+2e_{1}}_n|^{2}=\frac{8}{J_{0,n}^{2}}(1+ (v^{i}_n,v^{i+2e_{1}}_n)),
$$
by which
\begin{equation*}
(v_{n}^{i},v_{n}^{i+2})=\frac{J_{0,n}^{2}}{8}-1.
\end{equation*}
By this equality, using again \eqref{equality_{1}} we also get
\begin{equation*}
(v^{i}_n,v^{i+e_{1}}_n)=\frac{2}{J_{0,n}}(v^{i}_n,v^{i}_n+v^{i+2e_{1}}_n)=\frac{2}{J_{0,n}}(1+(v^{i}_n,v^{i+2e_{1}}_n))=\frac{J_{0,n}}{4}.
\end{equation*}
Since all the $v^{i}_{n}$ lie on a fixed plane the previous equality implies that $v_{n}$ agrees with the ground state $u_{n}$ defined in \eqref{ground} up to a fixed rotation $R\in SO(3)$.

The case $J_{0,n}>4$ trivially leads to ferromagnetic ground states (see also remark 3.3 in \cite{ciso}). 

As a result of this preliminary analysis, from now on we will focus on the asymptotics of the renormalized energy $H_{n}$. In particular, in what follows we consider the case when the parameter $J_{0,n}$ is in the vicinity of the Landau-Lifschitz point $J_{0}=4$ and the parameter $J_{2,n}$ diverges.  To this end we introduce $\dn\to 0$ and consider $J_{0,n}=4(1-\dn)$ so that $H_{n}$ takes now the form: 
\begin{align}\label{2d-normalized}
H_n(u):=\frac{1}{2}\left(\sum_{i\in R_n(\Omega)}\Ln^2|u^i-2(1-\dn)u^{i+e_1}+u^{i+2e_1}|^2+J_{2,n}\sum_{i\in R_n(\Omega)}\Ln^2|u^{i+e_2}-u^i|^2\right).
\end{align}
Within this choice stable states have a one dimensional helical structure and may exhibit chirality transitions in the propagation direction, which in our case is the horizontal axis. Consequently the analysis we are going to perform starts by considering energies on one dimensional horizontal slices of the domain. As we are going to show, this reduction to one dimensional spin chains still presents relevant differences with the case of $S^{1}$-valued spins considered in \cite{ciso}.

\subsection{One-dimensional slices}
In order to deal with one-dimensional energy slices we introduce the following additional notation. Let $I=(0,1)$ we define $\Z_n(I)$ as the set of those points $i \in \Z$ such that $\lambda_{n} i \in [0,1]$. We also define $R_n(I):=\{i\in \Z_n(I):\;i+2\in\Z_{n}(I)\}$. Similar to the two-dimensional case we will denote by ${\mathcal U}_{n}(I)$ the space of functions $u:i\in\Z_{n}(I)\mapsto u^{i}\in S^{2}$ and by $\overline{{\mathcal U}}_{n}(I)$ the subspace of those $u$ such that
\begin{equation}\label{prelim:boundary}
(u^{1},u^{0})=(u^{[1/\Ln]},u^{[1/\Ln]-1}).
\end{equation}
It is convenient to embed the family of configurations into a common function space. To this end we associate to any $u\in\overline{\mathcal{U}}_n(I)$ a piecewise-constant interpolation belonging to the class
\begin{equation}\label{cepsB}
C_n(I,S^{2}):=\{u\in\overline{\mathcal{U}}_n(I): u(x)=u(\Ln i)\,\,\hbox{ if } x \in \Ln(i+[0,1)),\,i \in \Z_n(I)\}. 
\end{equation}
The one-dimensional (sliced) renormalized energy is denoted by $H_{n}^{sl}:L^{\infty}(I,\R^{3})\to[0,+\infty]$ and takes the form below:
\begin{equation}\label{Hndelta}
H^{sl}_{n}(u):=\begin{cases}
\frac{1}{2}\sum_{i\in R_n(I)}\Ln\left|u^{i+2}-2(1-\dn)u^{i+1}+u^{i}\right|^{2}&\mbox{if $u\in C_n(I,S^{2})$,}\\
+\infty&\mbox{otherwise.}
\end{cases}
\end{equation} 
At first let us observe that the zero order $\Gamma$-limit is trivial. Indeed, the following result holds true.
\begin{proposition}\label{trivial-limit}
Let $H_{n}^{sl}:L^\infty(I,\R^{3})\to[0,+\infty]$ be the functional in \eqref{Hndelta}. Then $\Gamma\hbox{-}\lim_{n}H_{n}^{sl}(u)$ with respect to the weak-$*$ convergence in $L^{\infty}$ is given by 
\begin{equation*}
H^{sl}(u):=\begin{cases}
0&\mbox{if $|u|\leq 1$,}\\
+\infty&\mbox{otherwise in $L^{\infty}(I,\R^{3})$.}
\end{cases}
\end{equation*}
\end{proposition}
\proof
By \cite[Theorem 5.3]{ACG} there exists a convex function $f_{hom}:B(0,1)\to[0,+\infty)$ such that 
\begin{equation}
H^{sl}(u)=\begin{cases}
\int_{I}f_{hom}(u(x))\, \mathrm{d}x&\mbox{if $|u|\leq 1$,}\\
+\infty&\mbox{otherwise in $L^{\infty}(I,\R^{3})$.}
\end{cases}
\end{equation}
Let $u(x)\equiv u\in S^{2}$ be a constant function. Then, by a direct computation we have 
\begin{equation}
0\leq f_{hom}(u)=\int_{I}f_{hom}(u)\, \mathrm{d}x\leq\liminf_{n}H_{n}^{sl}(u)\leq \liminf_{n}2\delta_{n}^{2}=0.
\end{equation}
The result follows by the convexity of $f_{hom}$. 
\endproof
The degeneracy of the minima of $H^{sl}$ in the statement of Proposition \ref{trivial-limit} suggests to perform a higher order analysis by $\Gamma$-convergence in the spirit of \cite{Bra-Tru}. 

Let us recall a preliminary compactness result for scaled energies that was proved in \cite[Proposition 4.3]{ciso} for spin variables taking values in $S^1$ and whose proof works also for spins in $S^2$.

\begin{proposition}\label{compactness}
Let $\mu_n\to 0$ and let $u_n\in C_n(I,S^2)$ be such that 
\begin{equation}\label{prop:comp-bound}
\sup_n H_n^{sl}(u_n)\leq C\Ln\mu_n,
\end{equation}
then, for all $i$, we have
\begin{equation*}
\left|(1-\delta_{n})-(u_n^{i+1},u_n^i)\right|\leq C\mu_n^{\frac{1}{2}}.
\end{equation*}
In particular this implies that  $(u_n^{i+1},u_n^i)\to 1$ uniformly.
\end{proposition}

\section{$\Gamma$-convergence on slices}
This section is devoted to the study of the asymptotic behavior of the of one-dimensional renormalized energy \eqref{Hndelta}. We begin by introducing a convenient order parameter. Given a function $u\in C_{n}(I,S^{2})$, for all $i\in\{0,\dots,[\frac{1}{\Ln}]-1\}$ we set 
\begin{equation}\label{theta}
\theta^i(u)=\arccos((u^i,u^{i+1}))\in [0,\pi]
\end{equation} 
and ${w}^i=u^i\times u^{i+1}$. We now introduce a new order parameter $z:\Z_n(I) \rightarrow \R^3$ defined by
\begin{equation}\label{mapT}
z^i=\frac{1}{\sqrt{2\dn}}w^i=\frac{u^i\times u^{i+1}}{\sqrt{2\dn}}.
\end{equation}
which stands for a rescaled angular velocity. Such a $z$ will be extended in $L^1(I,\R^{3})$ by piecewise-constant interpolation. 
Note that the map $T_{n}:C_n(I,S^2)\rightarrow L^{1}(I,\R^3)$ associating to $u$ the corresponding $z$ according to \eqref{mapT} is not injective and that if $u$ satisfies periodic boundary conditions in the sense of \ref{prelim:boundary}, then $|z|$ is periodic and viceversa. As a result it can easily be seen that the energy cannot be uniquely defined by the function $z$. Therefore we define $H_n^{sl}$ on $L^1(I,\R^3)$ by setting
\begin{equation}\label{functionaldef}
H_n^{sl}(z)=
\begin{cases}
\inf_{T_{n}(u)=z}H^{sl}_n(u) &\mbox{if $z=T_{n}(u)\text{ for some }u\in C_n(I,S^2)$,}
\\
+\infty &\mbox{otherwise.}
\end{cases}
\end{equation}

\begin{remark}\label{T-1}
We stress that taking the infimum in the definition above has no effect in the asymptotic analysis we are going to perform. Indeed, if $u_{n},v_{n}$ are two sequences in $L^{\infty}(I,\R^{3})$ satisfying the energy bound \eqref{prop:comp-bound} and such that $z_{n}=T_{n}(u_n)=T_{n}(v_n)$, it easily follows from Proposition \ref{compactness} that for all $n$ large enough $n$, we have $(u_n^i,u_n^{i+1})=(v_n^i,v_n^{i+1})$ for all $i\in\{0,\dots,[\frac{1}{\Ln}]-1\}$. This also implies, by means of the identity 
\begin{equation*}
(u_n^i\times u_n^{i+1},u_n^{i+1}\times u_n^{i+2})=(u_n^i,u_n^{i+1})(u_n^{i+1},u_n^{i+2})-(u_n^i,u_n^{i+2})
\end{equation*}
that  $(u_n^{i},u_n^{i+2})=(v_n^i,v_n^{i+2})$ so that $H_{n}^{sl}(z_{n})$ does not depend on the element we choose in $T_{n}^{-1}(z_{n})$. 
\end{remark}

\subsection{General energy bounds}

As a preliminary result we point out some useful bounds on $H_{n}^{sl}$ at the energy scale $\Ln\delta_{n}^{3/2}$.

\begin{proposition}
Let $z_{n}$ be a sequence in $L^{\infty}(I,\R^{3})$ such that 
\begin{equation*}
\sup_n\frac{H_n^{sl}(z_n)}{\sqrt{2}\Ln\dn^{\frac{3}{2}}}\leq C<+\infty,
\end{equation*}
and let $u_{n}\in C_{n}(I,S^{2})$ be such that $z_{n}=T_{n}(u_{n})$ for all $n$. Then there exists a sequence of positive real numbers $\gamma_{n}\to 0$ such that for $n$ sufficiently large the following two bounds hold true:
\begin{eqnarray}
\frac{H^{sl}_n(z_n)}{\sqrt{2}\Ln\dn^{\frac{3}{2}}}\hspace{-.3cm}&\geq& \hspace{-.3cm}\frac{\sqrt{2\dn}}{\Ln}\sum_{i\in R_n(I)}\hspace{-.2cm}\Ln \left(\left|\frac{u_n^{i+1}-u_n^i}{\sqrt{2\dn}}\right|^2-1\right)^2+\frac{\Ln}{\sqrt{2\dn}}(1-\gamma_{n})\sum_{i\in R_n(I)}\hspace{-.2cm}\Ln\left|\frac{z_n^{i+1}-z_n^i}{\Ln}\right|^2\label{liminfbound}\\
\cr
\cr
\frac{H^{sl}_n(z_n)}{\sqrt{2}\Ln\dn^{\frac{3}{2}}}\hspace{-.3cm}&\leq& \hspace{-.3cm}\frac{\sqrt{2\dn}}{\Ln}\sum_{i\in R_n(I)}\hspace{-.2cm}\Ln \left(\left|\frac{u_n^{i+1}-u_n^i}{\sqrt{2\dn}}\right|^2-1\right)^2+\frac{\Ln}{\sqrt{2\dn}}\sum_{i\in R_n(I)}\hspace{-.2cm}\Ln\left|\frac{z_n^{i+1}-z_n^i}{\Ln}\right|^2.\label{limsupbound}
\end{eqnarray}
\end{proposition}
\proof
Since our assumption implies the energy bound \eqref{prop:comp-bound}, following remark \ref{T-1}, for $n$ sufficiently large the energy $H^{sl}_{n}(z_{n})$ can be rewritten in terms of $u_{n}\in T^{-1}(z_{n})$ and does not depend on the chosen element in $T^{-1}(z_{n})$. 
A straightforward calculations shows that
\begin{equation}\label{order4}
4|u_n^{i+1}-u_n^i|^2=|u_n^{i+1}-u_n^i|^4+4(1-(u_n^{i+1},u_n^i)^2).
\end{equation}
Thus we can rewrite the energy of $H_{n}^{sl}(z_{n})$ in terms of $u_{n}$ as 
\begin{align}\label{rewrite1}
H_n^{sl}(z_n)=&\sum_{i\in R_n(I)}2(1-\dn)\Ln|u_n^{i+1}-u_n^i|^2-\frac{\Ln}{2}|u_n^{i+2}-u_n^i|^2+2\Ln\dn^2\nonumber
\\
=&\sum_{i\in R_n(I)}\Ln\left(\frac{1}{2}|u_n^{i+1}-u_n^{i}|^4-2\dn|u_n^{i+1}-u_n^i|^2+2\dn^2\right)\nonumber
\\
&+\sum_{i\in R_n(I)}\Ln\left(2(1-(u_n^{i+1},u_n^i)^2)-\frac{1}{2}|u_n^{i+2}-u_n^i|^2\right)
\\
=&\sum_{i\in R_n(I)}2\Ln\left(\frac{1}{2}|u_n^{i+1}-u_n^i|^2-\dn\right)^2+\sum_{i\in R_n(I)}\Ln\left(2(1-(u_n^{i+1},u_n^i)^2)-\frac{1}{2}|u_n^{i+2}-u_n^i|^2\right)\nonumber
\\
=&2\delta_{n}^{2}\sum_{i\in R_n(I)}\Ln \left(\left|\frac{u_n^{i+1}-u_n^i}{\sqrt{2\dn}}\right|^2-1\right)^2
+\sum_{i\in R_n(I)}\Ln\left(2(1-(u_n^{i+1},u_n^i)^2)-\frac{1}{2}|u_n^{i+2}-u_n^i|^2\right)\nonumber.
\end{align}
We now claim that there exists a sequence of number $\gamma_{n}\to 0$ such that the following two inequalities hold:
\begin{eqnarray}
\sum_{i\in R_n(I)}\Ln\left(2(1-(u_n^{i+1},u_n^i)^2)-\frac{1}{2}|u_n^{i+2}-u_n^i|^2\right)&\geq& (1-\gamma_{n})\delta_{n}\sum_{i\in R_n(I)}\Ln\left|z_n^{i+1}-z_n^i\right|^2 \label{claim-lbound}\\
\sum_{i\in R_n(I)}\Ln\left(2(1-(u_n^{i+1},u_n^i)^2)-\frac{1}{2}|u_n^{i+2}-u_n^i|^2\right)&\leq& \delta_{n}\sum_{i\in R_n(I)}\Ln\left|z_n^{i+1}-z_n^i\right|^2 \label{claim-ubound}
\end{eqnarray}
If the claim is proved the inequalities in the statement follow by \eqref{rewrite1} on dividing by $\sqrt{2}\Ln\delta_{n}^{3/2}$.

We are then only left to show the validity of \eqref{claim-lbound} and \eqref{claim-ubound}. We first notice that by definition of $z_{n}$ we have 
\begin{equation}\label{equality-trivial}
1-(u_{n}^{i},u_{n}^{i+1})^{2}=2\dn|z_n^i|^2.
\end{equation}
Setting $\theta_{n}^{i}=\theta(u_{n})^{i}$ according to \eqref{theta} we observe that using the triple product expansion
\begin{align*}
u_n^{i+2}=&-w_n^{i+1}\times u_n^{i+1}+\cos(\theta_n^{i+1})\,u_n^{i+1},
\\
u_n^{i}=&w_n^{i}\times u_n^{i+1}+\cos(\theta_n^{i})\,u_n^{i+1}.
\end{align*}
Thus we can write 
\begin{align}\label{rodrigues}
|u_n^{i+2}-u_n^i|^2=&\left|(w_n^{i+1}+w_n^{i})\times u_n^{i+1}+(\cos(\theta_n^{i})-\cos(\theta_n^{i+1}))u_n^{i+1}\right|^2\nonumber
\\
=&\left|(w_n^{i+1}+w_n^{i})\times u_n^{i+1}\right|^2+\left|\cos(\theta_n^{i+1})-\cos(\theta_n^i)\right|^2\nonumber
\\
=&\left|w_n^{i+1}+w_n^{i}\right|^2+\left|\cos(\theta_n^{i+1})-\cos(\theta_n^i)\right|^2
\end{align} 
This immediately implies that 
\begin{equation}\label{upper-bound-trivial}
\frac12|u_n^{i+2}-u_n^i|^2\geq \delta_{n}|z_{n}^{i+1}+z_{n}^{i}|^{2}.
\end{equation}

By Proposition \ref{compactness} we have that $\theta_n^i\to 0$ uniformly in $i$. Combining that with the elementary fact that around zero $|\sin(x)|=\sin(|x|)$, it holds:
\begin{eqnarray*}
|\cos(\theta_n^{i+1})-\cos(\theta_n^i)|&=&|\cos(|\theta_n^{i+1}|)-\cos(|\theta_n^i|)|\nonumber\\
&\leq& \gamma_n||\sin(\theta_n^{i+1})|-|\sin(\theta_n^i)||\leq \gamma_n|w_n^{i+1}-w_n^i|
\end{eqnarray*}
for some sequence $\gamma_n$ converging to $0$. It then follows that
\begin{equation}\label{nnn}
\frac{1}{2}|u_n^{i+2}-u_n^i|^2\leq \delta_{n}\left(|z_n^{i+1}+z_n^i|^2+\gamma_n|z_n^{i+1}-z_n^i|^2\right).
\end{equation}
Inserting this estimate as well as \eqref{equality-trivial} in the left hand side of \eqref{claim-lbound} and using the periodicity of $|z_{n}|$ we have
\begin{align*}
\sum_{i\in R_n(I)}&\Ln\left(2(1-(u_n^{i+1},u_n^i)^2)-\frac{1}{2}|u_n^{i+2}-u_n^i|^2\right)
\\
&\geq\dn\sum_{i\in R_n(I)}\Ln\left(4|z_n^i|^2-|z_n^{i+1}+z_n^i|^2 -\gamma_n|z_n^{i+1}-z_n^i|^2\right)
\\
&=\dn\sum_{i\in R_n(I)}\Ln\left(2|z_n^i|^2+2|z_n^{i+1}|^2-|z_n^{i+1}+z_n^i|^2 -\gamma_n|z_n^{i+1}-z_n^i|^2\right)
\\
&=\dn(1-\gamma_n)\sum_{i\in R_n(I)}\Ln|z_n^{i+1}-z_n^i|^2.
\end{align*}
This proves claim \eqref{claim-lbound}. A similar argument using \eqref{upper-bound-trivial} in place of \eqref{nnn} proves claim \eqref{claim-ubound}.\endproof 
 
From the previous result we can deduce compactness with respect to the weak*-convergence in $L^{\infty}$. The bounds we find are indeed the best we can hope for in this case (see Remark \ref{noncompact} below), nevertheless they will play an important role in Section \ref{constrained} when we will discuss the coupling of the functional with other terms and we will use them in order to improve the compactness of sequences with equibounded energy.
The arguments are similar to the ones used in the proof of Theorem 1.2 in \cite{leoni}.

\begin{proposition}\label{weakcompactness}
Assume that $\frac{\Ln}{\sqrt{\dn}}\to 0$ and let $z_n\in L^{\infty}(I,\R^{3})$ be a sequence such that
\begin{equation*}
\sup_n \frac{H^{sl}_n(z_n)}{\sqrt{2}\Ln\dn^{\frac{3}{2}}}\leq C<+\infty.
\end{equation*}
Then $\|z_n\|_{\infty}$ is equibounded and, up to subsequences, $z_n$ converges weakly* in $L^{\infty}(I)$ to some $z\in L^{\infty}(I,B(0,1))$. If in addition $z_{n}\to z$ in $L^{1}(I)$, then $z\in L^{\infty}(I,S^{2})$.
\end{proposition}

\begin{proof}
Let $u_{n}$ be such that $T_{n}(u_{n})=z_{n}$. By \eqref{liminfbound} we have that, for large $n$,
\begin{align*}
C&\geq\frac{2\sqrt{\dn}}{\Ln}\sum_{i\in R_n(I)}\Ln\left(\left|\frac{u_n^{i+1}-u_n^i}{\sqrt{2\dn}}\right|^2-1\right)^2+\frac{\Ln}{2\sqrt{\dn}}\sum_{i\in R_n(I)}\Ln\left|\frac{z_n^{i+1}-z_n^i}{\Ln}\right|^2
\\
&\geq 2\sum_{i\in R_n(I)}\Ln\left|\left|\frac{u_n^{i+1}-u_n^i}{\sqrt{2\dn}}\right|^2-1\right|\left|\frac{z_n^{i+1}-z_n^i}{\Ln}\right|.
\end{align*}
First we observe that 
\begin{equation}\label{dist-graduz}
\left|\frac{u_n^{i+1}-u_n^i}{\sqrt{2\dn}}\right|^2-|z_n^i|^2=\frac{\left((u_n^{i+1},u_n^i)-1\right)^2}{2\dn}\geq 0.
\end{equation}
As a result we can continue the lower bound above deducing that 
\begin{equation}\label{thresholdest}
C\geq 6\sum_{i\in R_n(I):\,|z_n^i|\geq 2}\Ln\left|\frac{z_n^{i+1}-z_n^i}{\Ln}\right|.
\end{equation}
Exploiting again \eqref{liminfbound}, we may also deduce that
\begin{equation}\label{continuity}
\sup_i|z_n^{i+1}-z_n^i|^2\leq 2C\sqrt{\dn}.
\end{equation}
We now fix $n_{0}$ such that $2C\sqrt{\dn}<\frac14$ for all $n>n_0$. Given $n>n_{0}$, we claim that $\|z_{n}\|_{\infty}\leq\max\{4, 3+\frac{C}{6}\}$. To this end assume it exists $j$ such that $|z_{n}^{j}|\geq 4$, otherwise the claim is proved.

We observe that, combining \eqref{liminfbound} with \eqref{dist-graduz}, there exists $i(n)\in R_n(I)$ such that $|z_n^{i(n)}|^2\leq 2$. Without loss of generality we may suppose that $i(n)<j$. Let us define 
$$
k(n)+1:=\min\{i:\, i(n)\leq i\leq j \,{ \rm with }\,|z_{n}^{l}|> 3,\, \forall\, i\leq l\leq j\}.
$$ 
The minimum is well defined since the set contains at least $j$. Note that $k(n)+1> i(n)$, that gives $k(n)\geq 0$ and $|z_{n}^{k(n)}|\leq 3$. By \eqref{continuity} and the choice of $n>n_{0}$ we also have that $|z_{n}^{k(n)}|\geq 2$. Therefore we have that $|z_{n}^{l}|\geq 2$ for all $k(n)\leq l\leq j$ and by \eqref{thresholdest} we eventually have that
\begin{equation*}
|z_n^j|\leq |z_n^{k(n)}|+\sum_{l=k(n)}^{j-1}\Ln\left|\frac{z_n^{l+1}-z_n^l}{\Ln}\right|\leq 3+\frac{C}{6},
\end{equation*}
which proves the claim. As a result equiboundedness as well as $L^{\infty}$ weak* compactness are shown. 

We now prove that $|z_{n}|\to 1$ almost everywhere in $I$. Setting $\theta^{i}_{n}=\theta^{i}(u_{n})$ according to \eqref{theta} we define the piecewise constant function ${\zeta}_n$ whose value on the nodes of the lattice is
\begin{equation*}
{\zeta}_n^i=\sqrt{\frac{2}{\dn}}\left|\sin\left(\frac{\theta_n^i}{2}\right)\right|.
\end{equation*}
We notice that for all $i\in R_{n}(I)$ one has by definition that $\zeta^{i}_{n}=\left|\frac{u_{n}^{i+1}-u_{n}^{i}}{\sqrt{2\dn}}\right|$.
Since by Proposition \ref{compactness} $\theta^{i}_{n}\to 0$ uniformly in $i$, by the equiboundedness of $|z^{i}_{n}|=|\frac{\sin(\theta_{n}^{i})}{\sqrt 2\dn}|$ and trigonometric identities we get that $\||z_{n}|-\zeta_{n}\|_{\infty}\to 0$. By \eqref{liminfbound} we may now write that, for any interval $I^{\prime}\subset\subset I$,
\begin{align*}
0=\lim_{n}C\frac{\Ln}{2\sqrt{\dn}}\geq\limsup_{n}\sum_{i\in R_n(I)}\Ln\left(\left|\frac{u_n^{i+1}-u_n^i}{\sqrt{2\dn}}\right|^2-1\right)^2\geq \limsup_{n}\int_{I^{\prime}}\left({\zeta}_n(t)^2-1\right)^2\,\mathrm{d}t.
\end{align*}
Therefore ${\zeta}_n\to 1$ pointwise almost everywhere due to the arbitrariness of $I^{\prime}$ which implies that $|z_{n}|\to 1$ almost everywhere.
It follows that any weak limit of $z_{n}$ belongs to $L^{\infty}(I,B(0,1))$ and that if in addition $z_{n}\to z$ strongly in $L^{1}$ then $z\in L^{\infty}(I, S^{2})$. 
\end{proof}

\begin{remark}\label{affineclose}
If $z_n$ is as in Proposition \ref{weakcompactness} and $z_n^a$ denotes the piecewise affine interpolation on $\mathbb{Z}_n(I)$ of $z_n$, then it follows from (\ref{continuity}) that

\begin{equation*}
\sup_{t\in I}|z_n^a(t)-z_n(t)|^2\leq 2C\sqrt{\dn}.
\end{equation*}
This estimate of course also holds if we rescale the variable $t$.
\end{remark}

\subsection{Zero energy chirality transitions:}
In this section we will prove that, in contrast to the $S^{1}$-valued spin system studied in \cite{ciso}, in the present case the functional $H_{n}^{sl}$ does not penalize chirality transitions between ground states. In other words the optimal asymptotic energy for a transition turns out to be zero, as it is explained below. Before entering into the details of the proof we need to introduce some notation. Given two unit vectors $z_{-},z_{+}\in S^2$ we set
\begin{equation*}
H^{\times}_{z_{-},z_{+}}:=\left\{w=u\times u^{\prime},\; u\in H^2_{loc}(\R,S^2):\;\lim_{t\to \pm\infty}w(t)=z_{\pm}\right\}.
\end{equation*}
We first prove the following lemma.
\begin{lemma}\label{rightspace}
Let $u\in H^2_{loc}(\R,S^2)$ and let $w=u\times u^{\prime}$. Then $w\in H^1_{loc}(\R,\R^3)$ and $w^{\prime}=u\times u^{\prime\prime}$. 
\\
Moreover, if $u\in H^1_{loc}(\R,S^2)$ and $w=u\times u^{\prime}\in H^1_{loc}(\R,\R^3)$, then $u\in H^2_{loc}(\R,S^2)$.
\end{lemma}

\begin{proof}
The first statement can be proved by approximation with smooth functions. Concerning the second one, note that
\begin{equation*}
u\times w= u\times (u\times u^{\prime})=(u,u^{\prime})u-(u,u)u^{\prime}=-u^{\prime},
\end{equation*}
where we have used that $|u|=1$, so that $(u,u^{,\prime})=0$ almost everywhere. On every bounded interval $J$ we have $u,w\in H^1(J,\R^3)\cap L^{\infty}(J,\R^3)$, so that $u\times w\in H^1(J,\R^3)$.
\end{proof}
We define the transition energy function $g:S^2\times S^2\rightarrow [0,+\infty)$ by
\begin{equation}\label{transitionenergy}
g(z_1,z_2):=\inf\left\{\int_{\R}(|w(t)|^2-1)^2\,\mathrm{d}t+\int_{\R}|w^{\prime}(t)|^2\,\mathrm{d}t:\;w\in H^{\times}_{z_1,z_2}\right\}.
\end{equation}
In the following lemma we show that actually the infimum is zero for every $z_1, z_2\in S^{2}$.
\begin{lemma}\label{bvbound}
For all $z_1,z_2\in S^2$ we have
\begin{equation*}
g(z_1,z_2)=0.
\end{equation*}
\end{lemma}
\begin{proof}
The function $g$ is invariant under rotations so we may assume that $z_1=e_3$ and $z_2=\lambda e_2+\mu e_3$ with $\lambda^2+\mu^2=1$. Now we take a $C^2$ cut-off function $\gamma:\R\rightarrow [0,1]$ such that
\begin{equation*}
\gamma(t)=
\begin{cases}
0 &\mbox{$t\leq 0$,}\\
1 &\mbox{$t\geq 1$}.
\end{cases}
\end{equation*}
Given $\rho>1$, we define $\gamma_{\rho}:\R\rightarrow [0,1]$ as $\gamma_{\rho}(t)=\gamma(\frac{t}{\rho})$. We consider the matrix
$$A=\begin{pmatrix}
1 &0 &0 \\
0 &\mu &\lambda \\
0 &-\lambda &\mu
\end{pmatrix}$$
which belongs to $SO(3)$ and maps $z_1$ to $z_2$. Let $B$ be an antisymmetric matrix such that $A=\exp(B)$. We define the test function in the infimum problem defining $g$ by
\begin{equation}\label{testinf}
u_{\rho}(t):=\exp(\gamma_{\rho}(t)B)(\cos(t),\sin(t),0).
\end{equation}
Then $u_{\rho}\in H^2_{loc}(\R,S^2)$ and, since $B$ commutes with $\exp(\gamma_{\rho}(t)B)$,
\begin{equation}\label{derivative}
u_{\rho}^{\prime}(t)=\gamma_{\rho}^{\prime}(t)B\,u_{\rho}(t)+\exp(\gamma_{\rho}(t)B)(-\sin(t),\cos(t),0).
\end{equation}
Since $\gamma_{\rho}^{\prime}(t)=0$ for $t\notin (0,\rho)$ it follows that $w_{\rho}=u_{\rho}\times u_{\rho}^{\prime}$ satisfies
\begin{equation}\label{test}
w_{\rho}(t)=\begin{cases}
e_{3} &\mbox{$t\leq 0$,}\\
\lambda e_{2}+\mu e_{3} &\mbox{$t\geq \rho$}.
\end{cases}
\end{equation}
By Lemma \ref{rightspace}, $w_{\rho}\in H^1_{loc}(\R)$ so that $w_{\rho}\in H^{\times}_{z_{1},z_{2}}$. Moreover from (\ref{derivative}) it follows that there exists a constant $C$ depending only on $|B|$ and on the $C^{2}$-norm of $\gamma$ in $[0,1]$, such that
\begin{equation*}
1-C{\rho}^{-1}\leq |u^{\prime}_{\rho}(t)|\leq 1+C{\rho}^{-1} \quad\mbox {if $t\in (0,\rho)$.}
\end{equation*}
Taking squares in the previous inequality and since $|w_{\rho}(t)|=|u_{\rho}^{\prime}(t)|$ we deduce that 
\begin{equation}\label{energycontributon1}
\left(|w_{\rho}(t)|^2-1\right)^2\leq C\rho^{-2}.
\end{equation}
Since the second derivative of $u_{\rho}$ reads as
\begin{equation*}
u_{\rho}^{\prime\prime}(t)=\gamma_{\rho}^{\prime\prime}(t)B\,u_{\rho}(t)+\gamma_{\rho}^{\prime}(t)B\,u_{\rho}^{\prime}(t)\\
+\exp(\gamma_{\rho}(t)B)\gamma_{\rho}^{\prime}(t)B(-\sin(t),\cos(t),0)-u_{\rho}(t).
\end{equation*}
we infer that
\begin{align}\label{derivativecontr}
|w_{\rho}^{\prime}(t)|^2=|u_{\rho}(t)\times u_{\rho}^{\prime\prime}(t)|^2\leq C\rho^{-2}\quad\mbox {if $t\in (0,\rho)$.}
\end{align}
For $t\notin (0,\rho)$ by \eqref{test}$ w_{\rho}(t)$ does not contribute to (\ref{transitionenergy}). It then follows from (\ref{energycontributon1}) and (\ref{derivativecontr}) that
\begin{equation*}
g(z_1,z_2)\leq C\rho^{-1},
\end{equation*}
which implies $g(z_1,z_2)=0$ by the arbitrariness of $\rho$.
\end{proof} 
We are now going to compute the $\Gamma$-limit of $H_{n}^{sl}$ with respect to the weak$^*$ convergence where we have proved a compactness result (see Proposition \ref{weakcompactness}). First notice that this choice forces us to restrict the domain of the functional to some a priori fixed ball of $L^{\infty}$ where the weak$^{*}$ topology is metrizable. On the other hand, as it will be clear from our $\Gamma$-limsup construction, without the addition of other terms to the functional, there is no hope for compactness in a finer topology. 

The following lemma will be used in the proof of the next theorem as well as in the sequel of the paper.
\begin{lemma}\label{Rodriguez}
Let $u\in C^{1}((a,b),S^{2})$ and $w\in S^2$ be such that $u(s)\times u'(s)=w$ for all $s\in(a,b)$. Then, for all $s_{1},s_{2}\in(a,b)$ it holds
\begin{equation}
u(s_{1})\times u(s_{2})=\sin({s_{2}-s_{1}})\, w.
\end{equation}
\end{lemma}
\proof The result follows from a direct computation.\qed

\smallskip

For every $R> 1$ we define the functional $H_{n}^{sl,R}$ as follows:
\begin{equation}\label{HnslR}
H_{n}^{sl,R}(z)=
\begin{cases}
H_{n}^{sl}(z) &\mbox{if $\|z\|_{\infty}\leq R$,}\\
+\infty &\mbox{otherwise.}
\end{cases}.
\end{equation}
The following $\Gamma$-convergence result holds.
\begin{theorem}\label{main1}
Let $R>1$ and $H_n^{sl,R}:L^1(I,\R^3)\rightarrow [0,+\infty]$ be defined as in (\ref{HnslR}). Assume that $\frac{\Ln}{\sqrt{\dn}}\to 0$. Then the functionals $\frac{H_n^{sl,R}}{\sqrt{2}\Ln\dn^{\frac{3}{2}}}$ $\Gamma$-converge with respect to the weak$^{*}$ $L^{\infty}$-convergence to the functional 
\begin{equation*}
H^{sl}(z)=
\begin{cases}
0 &\mbox{if $z\in L^{\infty}(I,B(0,1))$,}\\
+\infty &\mbox{otherwise.}
\end{cases}.
\end{equation*}
\end{theorem}
\begin{proof}
\noindent\textbf{liminf-inequality.} Since $H_{n}^{sl,R}\geq 0$ it only suffices to check that any weak$^{*}$ limit of sequences $z_{n}$ such that $\sup_{n}\frac{H_n^{sl,R}(z_{n})}{\sqrt{2}\Ln\dn^{\frac{3}{2}}}\leq C<+\infty$ belongs to $L^{\infty}(I,S^2)$. This is ensured by Proposition \ref{weakcompactness}. 

\textbf{limsup-inequality:} By density it suffices to prove the inequality for a $S^{2}$-valued piecewise constant function $z$. Since the construction of the recovery sequence will be local we can assume that $z=z_1\mathds{1}_{[0,\frac{1}{2})}+z_2\mathds{1}_{(\frac{1}{2},1]}$ with $|z_{1}|=|z_{2}|=1$. Given $\e>0$ we find a function $u\in H^2_{loc}(\R,S^2)$ such that $w=u\times u^{\prime}$ is admissible in the infimum problem defining $g(z_1,z_2)$ in (\ref{transitionenergy}) and
\begin{equation}\label{infappr}
\int_{\R}(|w(t)|^2-1)^2\dt+\int_{\R}|w^{\prime}(t)|^2\dt\leq \e.
\end{equation}
Having in mind the family constructed in the proof of Lemma \ref{bvbound} we can further assume that $u\in C^2(\R,S^2)$, that it has bounded and uniformly continuous first and second derivative, that it satisfies the bound $\|u'\|_{\infty}\leq 1+\e$ and that there exists $t_{\e}>0$ such that
\begin{align}
w(t)&=z_1 \quad\forall t\leq 0,\label{at-inf}\\
w(t)&=z_2 \quad\forall t\geq t_{\e},\label{at+inf}.
\end{align}
We consider the sequence $\alpha_{n}=\frac{\arccos(1-\dn)}{\sqrt{2\dn}}$ and we observe that $1-\dn=\cos(\alpha_{n}\sqrt{2\dn})$ and that $\lim_{n}\alpha_{n}=1$. We now define the function $u_n\in C_n(I,S^2)$ setting
\begin{equation*}
u_n^i=u\left(\alpha_{n}\frac{\sqrt{2\dn}}{\Ln}(\Ln i-\frac{1}{2})\right).
\end{equation*}
By the uniform continuity of $u$ we have that, for large enough $n$, for $j\in\{1,2\}$ $|u^{i+j}_{n}-u_{n}^{i}|\to 0$ uniformly with respect to $i\in R_{n}(I)$. In particular this implies that 
\begin{equation}\label{scalarpos}
(u_{n}^{i+j},u_{n}^{i})>0,\ {\rm for }\,j\in\{1,2\}.
\end{equation}
We now fix $i_{-}=\left[\frac{1}{2\Ln}\right]$ and $i_{+}=\left[\frac{1}{2\Ln}+\frac{t_{\e}}{\alpha_{n}\sqrt{2\dn}}\right]+1$ and observe that $i_{\pm}\in R_{n}(I)$ for $n$ large enough. Applying Lemma \ref{Rodriguez} to the function $u$ in the interval $(-\infty,\alpha_{n}\frac{\sqrt{2\dn}}{\Ln}(\Ln i_{-}-\frac{1}{2}))$ we get that for all $i<i_{-}-1$ it holds that 
\begin{equation*}
u_{n}^{i}\times u_{n}^{i+1}=\sin\left(\alpha_{n}{\sqrt{2\dn}}\right) z_{1}, \quad u_{n}^{i}\times u_{n}^{i+2}=\sin\left(2\alpha_{n}{\sqrt{2\dn}}\right) z_{1}.
\end{equation*}
Using \eqref{scalarpos} we get
\begin{equation}\label{ground-1}
(u_{n}^{i},u_{n}^{i+1})=\cos\left(\alpha_{n}\sqrt{2\dn}\right)=1-\dn, \quad (u_{n}^{i}, u_{n}^{i+2})=\cos\left(2\alpha_{n}{\sqrt{2\dn}}\right)=2(1-\dn^{2})-1
\end{equation}
which further implies 
\begin{equation}\label{ground-2}
u_{n}^{i}-2u_{n}^{i+1}+u_{n}^{i+2}=2\dn u_{n}^{i+1}.
\end{equation}
Using the same argument in the interval $(\alpha_{n}\frac{\sqrt{2\dn}}{\Ln}(\Ln i_{+}-\frac{1}{2}),+\infty)$ with $z_{2}$ in place of $z_{1}$ it can be shown that \eqref{ground-1} and \eqref{ground-2} hold for all $i\geq i_{+}$. In particular the first equality in \eqref{ground-1} implies that $u_{n}$ satisfies the boundary conditions in \eqref{prelim:boundary}. 

We now consider the sequence $z_n:=T_{n}(u_n)$. It holds that for all $i\in R_{n}(I)$
\begin{equation}\label{boundzn}
|z_n^i|=|u_n^i\times\frac{u_n^{i+1}-u_n^i}{\sqrt{2\dn}}|\leq |\frac{u_n^{i+1}-u_n^i}{\sqrt{2\dn}}|\leq \alpha_{n}\|u^{\prime}\|_{\infty}\leq \alpha_{n}(1+\e).
\end{equation}
Furthermore, since the first derivative of $u$ is uniformly continuous on $\R$, $z_n$ converges pointwise almost everywhere to $z$ and on applying dominated convergence also in $L^1$ as well as in the weak-$^{*}$ convergence of $L^{\infty}$ by \eqref{boundzn}. Thus $z_n$ is an admissible recovery sequence. 
We now define the auxiliary functions $\tilde{z}_n:\R\rightarrow\R^3$ by
\begin{equation*}
\tilde{z}_n(s)=
\begin{cases}
\frac{u_n^{i+1}-u_n^i}{\sqrt{2\dn}} &\mbox{if $s\in [\frac{\alpha_{n}\sqrt{2\dn}}{\Ln}(\Ln i-\frac{1}{2}),\frac{\alpha_{n}\sqrt{2\dn}}{\Ln}(\Ln (i+1)-\frac{1}{2})),\;i\in{0,\dots,[\frac{1}{\Ln}]-2}$,}
\\
u'(s) &\mbox{otherwise.}
\end{cases}
\end{equation*}
By the change of variables $s-\frac{1}{2}=\frac{\Ln}{\alpha_{n}\sqrt{2\dn}}t$ we have
\begin{align}\label{limsup0}
\frac{\sqrt{2\dn}}{\Ln}&\sum_{i=0}^{\left[1/\Ln\right]-2}\Ln \left(\left|\frac{u_n^{i+1}-u_n^i}{\sqrt{2\dn}}\right|^2-1\right)^2\leq \frac{\sqrt{2\dn}}{\Ln}\int_0^1\left(\left|\tilde{z}_n\left(\frac{\alpha_{n}\sqrt{2\dn}}{\Ln}(s-\frac{1}{2})\right)\right|^2-1\right)^2\,\mathrm{d}s\nonumber\\
&=\frac{1}{\alpha_{n}}\int_{\frac{-\alpha_{n}\sqrt{\dn}}{\sqrt{2}\Ln}}^{\frac{\alpha_{n}\sqrt{\dn}}{\sqrt{2}\Ln}}\left(|\tilde{z}_n(t)|^2-1\right)^2\dt\leq\frac{1}{\alpha_{n}}\int_{\R}\left(|\tilde{z}_n(t)|^2-1\right)^2\dt.
\end{align}
Since $(-\infty,-\frac12)\subset(-\infty,\alpha_{n}\frac{\sqrt{2\dn}}{\Ln}(\Ln i_{-}-\frac{1}{2}))$ and $(t_{\e}+1,+\infty)\subset(\alpha_{n}\frac{\sqrt{2\dn}}{\Ln}(\Ln i_{+}-\frac{1}{2}))$, using \eqref{at-inf}, \eqref{at+inf} and \eqref{ground-1} one has that 
\begin{equation}\label{domconv}
|\tilde{z}_{n}(s)|=1\quad \forall\, s\in(-\infty,-\frac12)\cup(t_{\e}+1,+\infty).
\end{equation}
Since $u^{\prime}$ is uniformly continuous and $\alpha_{n}\to 1$ it is easy to see that $\tilde{z}_n$ converges uniformly to $u^{\prime}$ on $[-\frac12,t_{\e}+1]$. By \eqref{domconv} we can apply dominated convergence in the r.h.s of \eqref{limsup0}. Since $|u'(s)|=|w(s)|$ for all $s\in\R$, we deduce 
\begin{equation}\label{limsup1}
\limsup_{n\to\infty}\frac{\sqrt{2\dn}}{\Ln}\sum_{i=0}^{\left[1/\Ln\right]-2}\Ln \left(\left|\frac{u_n^{i+1}-u_n^i}{\sqrt{2\dn}}\right|^2-1\right)^2\leq\int_{\R}\left(|w(t)|^2-1\right)^2\dt.
\end{equation}

We now define the auxiliary function $\overline{z}_n:\R\rightarrow\R^3$ as
\begin{equation*}
\overline{z}_n(s)=
\begin{cases}
\frac{z_n^{i+1}-z_n^i}{\sqrt{2\dn}} &\mbox{if $s\in [\frac{\alpha_{n}\sqrt{2\dn}}{\Ln}(\Ln i-\frac{1}{2}),\frac{\alpha_{n}\sqrt{2\dn}}{\Ln}(\Ln (i+1)-\frac{1}{2})),\;i\in{0,\dots,[\frac{1}{\Ln}]-2}$,}
\\
w^{\prime}(s) &\mbox{otherwise.}
\end{cases}
\end{equation*}
Using again the same change of variables as above we get
\begin{align}
\frac{\Ln}{\sqrt{2\dn}}&\sum_{i=0}^{\left[1/\Ln\right]-2}\Ln\left|\frac{z_n^{i+1}-z_n^i}{\Ln}\right|^2\leq \frac{\sqrt{2\dn}}{\Ln}\sum_{i=0}^{\left[1/\Ln\right]-2}\Ln\left|\frac{z_n^{i+1}-z_n^i}{\sqrt{2\dn}}\right|^2\nonumber\\
&=\frac{\sqrt{2\dn}}{\Ln}\int_0^1\left|\overline{z}_n\left(\frac{\alpha_{n}\sqrt{2\dn}}{\Ln}(s-\frac{1}{2})\right)\right|^2\,\mathrm{d}s=\frac{1}{\alpha_{n}}\int_{\frac{-\alpha_{n}\sqrt{\dn}}{\sqrt{2}\Ln}}^{\frac{\alpha_{n}\sqrt{\dn}}{\sqrt{2}\Ln}}|\overline{z}_n(t)|^2\dt.\label{limitder}
\end{align}
We first claim that $\overline{z}_n$ converges to $w^{\prime}$ in $L^2_{loc}(\R)$. To see this let us set $h_n=\alpha_{n}\sqrt{2\dn}$ and note that
\begin{align*}
\overline{z}_n(s)&=\frac{z_n^{i+1}-z_n^i}{\sqrt{2\dn}}=\alpha_{n}^{2}\,u_n^{i+1}\times \frac{(u_n^{i+2}-2u_n^{i+1}+u_n^i)}{h_n^2}
\\
&=\alpha_{n}^{2}\, u(\xi_n^i)\times\frac{u(\xi_n^i+h_n)-2u(\xi_n^i)+u(\xi_n^i-h_n)}{h_n^2},
\end{align*}
for some $|\xi_n^i-s|\leq h_n$. By the continuity of $u^{\prime\prime}$, since $h_n\to 0$ and $\alpha_{n}\to 1$, we have $\overline{z}_n(s)\to u(s)\times u^{\prime\prime}(s)=w^{\prime}(s)$ which proves the pointwise convergence of $\overline z_{n}$ to $w'$. Since $u''$ is uniformly bounded on $\R$ it follows that $\overline{z}_n$ is equibounded which gives the $L^{2}_{loc}$ convergence. On the other hand, thanks to \eqref{ground-2}, we have that $\overline z_{n}(s)=0$ for all $s\in(-\infty,-\frac12)\cup(t_{\e}+1,+\infty)$. Therefore we can let $n\to +\infty$ in (\ref{limitder}) and deduce
\begin{equation}\label{limsup2}
\limsup_{n\to\infty}\frac{\Ln}{\sqrt{2\dn}}\sum_{i=0}^{\left[1/\Ln\right]-2}\Ln\left|\frac{z_n^{i+1}-z_n^i}{\Ln}\right|^2\leq\int_{\R}|w^{\prime}(t)|^2\dt.
\end{equation}
Combining \eqref{limsupbound}, (\ref{infappr}), (\ref{limsup1}) and (\ref{limsup2}), by the arbitrainess of $\e$ we infer that
\begin{equation}
\Gamma-\limsup_{n\to\infty}\frac{H^{sl}_n(z)}{\sqrt{2}\Ln\dn^{\frac{3}{2}}}\leq 0.
\end{equation}
\end{proof}

\begin{remark}\label{noncompact}
Assume that $\frac{\Ln}{\sqrt{\dn}}\to 0$. Then there exists a sequence of functions $z_n\in C_n(I,\R^3)$ such that
\begin{equation*}
H^{sl}_n(z_n)\leq C\Ln\dn^{\frac{3}{2}}
\end{equation*}
such that no subsequence converges strongly in $L^1(I,\R^3)$. In fact, let us fix $\eta_{n}=c_{n}\Ln$ where $c_{n}\in\N$ is such that $\frac{\Ln}{\sqrt{2\dn}}<<\eta_{n}<<1$. Let us consider $\overline u_{n}\in H^{2}_{loc}(\R,S^{2})$ such that $w_{n}=u_{n}\times u_{n}'$ with $w_{n}$ satisfying the properties \eqref{at-inf} and \eqref{at+inf} with $z_{2}=-z_{1}$ and such that \eqref{infappr} holds with $\eta_{n}^{2}$ in place of $\e$. For all $i\in\{0,\dots,\frac{2\eta_{n}}{\Ln}\}$ we set 
\begin{equation*}
u_n^i=\overline u_{n}\left(\alpha_{n}\frac{\sqrt{2\dn}}{\Ln}(\Ln i-\eta_{n})\right)
\end{equation*}
and we define $u_{n}\in C_{n}(I,S^{2})$ as the $2\eta_{n}$-periodic extension of the function above. Setting $z_{n}=T_{n}(u_{n})$ by construction we have that $z_{n}\to 0$ in the weak$^{*}$ topology of $L^{\infty}$. By repeating the same argument in the proof of the $\Gamma$-limsup inequality, the energy stored in each interval of length $\eta_{n}$ is at most $\eta_{n}^{2}$, so that
$$
\frac{H_{n}^{sl}(z_{n})}{\sqrt{2}\Ln\dn^{\frac32}}\leq\frac{\eta^{2}_{n}}{\eta_{n}}\to 0.
$$
The sequence constructed in this way cannot converge strongly to $z=0$, otherwise by Proposition \ref{weakcompactness} we would get $z\in L^1(I,S^2)$.
\end{remark}

\subsection{$S^{2}$-chirality transitions under additional constraints}\label{constrained}
As discussed in the previous section, it is not possible to energetically detect chirality transitions by using as energy $H_{n}$, that is the scaled NN and NNN frustrated spin chain model as in \cite{DmiKri}. Nevertheless, transitions with non trivial energy may appear if we modify the functional $H_{n}$ by adding what we call either a {\it hard} or a {\it soft} penalization term. In the {\it hard} case we will force the spin variable to take values only in a subset of $S^{2}$ consisting of finitely many copies of $S^1$, while in the {\it soft} case we will penalize the distance of the spin field from such a set. The main difference between the two cases is that, while in the first case we will prove that chirality transitions leads to a constant positive limit energy to be paid for each discontinuity in the chirality, in the second one we can present some examples showing dependence of the limit energy on the two chiral states between which the transition occurs.

\subsubsection{Chirality transitions via {\it hard} penalization}\label{hard}

Let $q_1,\dots,q_k$ be a fixed family of distinct points in $S^2$, where $k\geq 1$. For $l\in\{1,2,\dots,k\}$ we set $S^1_l=S^{2}\cap q_{l}^{\perp}$. To reduce notation we set
\begin{align}\label{Hpen}
&Q_k:=\{\pm q_1,\dots,\pm q_k\},\quad
M_k:=\bigcup_{l=1}^k S^1_l,\quad
L_k:=\bigcup_{l=1}^k\text{span}(q_l).
\end{align}
We then restrict the spin variable $u$ to take values only in $M_k$. We define the space $C_{n}(I,M_{k})$ as the subset of $C_{n}(I,S^{2})$ of those functions taking values in $M_{k}$. We define the energy \hbox{$H^{sl,k}_n:L^1(I)\rightarrow [0,+\infty]$} as
\begin{equation}\label{newfunctional}
H^{sl,k}_n(z)=
\begin{cases}
\inf_{T(u)=z}H^{sl}_n(u) &\mbox{if $z=T_{n}(u)$ for some $u\in C_n(I,M_{k})$,}\\
+\infty &\mbox{otherwise.}
\end{cases}
\end{equation}
Moreover we set
\begin{equation*}
H^{\times}_{q_{-},q_{+}}(M_k):=\left\{w=u\times u^{\prime},\; u\in H^2_{loc}(\R,M_k):\;\lim_{t\to \pm\infty}w(t)=q_{\pm}\right\}
\end{equation*}
and define the function $h_{k}:Q_{k}\times Q_{k}\rightarrow\R$ by
\begin{equation}\label{newtransitionenergy}
h_{k}(q_{-},q_{+}):=\inf\left\{\int_{\R}(|w(t)|^2-1)^2\dt+\int_{\R}|w^{\prime}(t)|^2\dt:\;w\in H^{\times}_{q_{-},q_{+}}(M_k)\right\}.
\end{equation}
In this setting the function $h_{k}$ turns out to be independent of the $k$ as well as of $(q_{-},q_{+})\in Q_{k}\times Q_{k}$ and reduces to the well-known transition energy for scalar problems as shown in the next lemma.
\begin{lemma}\label{transitionvalue}
Let $q_{-},q_{+}\in Q_k,\;q_{-}\neq q_{+}$. Then $h_{k}(q_{-},q_{+})=\frac{8}{3}$ and we have equivalently
\begin{equation*}
h_k(q_{-},q_{+})=\inf\left\{\int_{\R}(|w(t)|^2-1)^2+|w^{\prime}(t)|^2\dt:\;w\in H^1_{loc}(\R,L_k),\,\lim_{t\to\pm\infty}w(t)=q_{\pm}\right\},
\end{equation*} 
which is solved by the function $w_{q_{-},q_{+}}$ defined as
\begin{equation*}
w_{q_{-},q_{+}}(t)=
\begin{cases}
|\tanh(t)|q_{-} &\mbox{if $t\leq 0$,}\\
|\tanh(t)|q_{+} &\mbox{if $t>0$.}
\end{cases}
\end{equation*}
\end{lemma}

\begin{proof}
We first show that $w_{q_{-},q_{+}}$ is the solution of the minimum problem if we replace the cross product constraint by requiring $w\in H^1(\R,L_k)$. To this end (taking a continuous representative) note that we don't increase the energy if we stay in the half line $g_{q_{+}}:=\{\lambda\,q_{+}:\;\lambda\geq 0\}$ as soon as we reach the origin for the first time coming from $q_{-}$. Indeed, if $t_0=\inf\{t\in\R:\;w(t)=0\}$ and $t_1=\sup\{t\in\R:\;w(t)=0\}$, then the function
\begin{equation}\label{cutout}
\tilde{w}(t)=
\begin{cases}
w(t) &\mbox{if $t<t_0$,}\\
w(t-t_0+t_1) &\mbox{if $t\geq t_0$}
\end{cases}
\end{equation} 
gives the same or less energy as $w$. Now given such a function $w$ we define $v\in H^1(\R)$ setting
\begin{equation*}
v(t)=
\begin{cases}
-|w(t)| &\mbox{if $w(t)\in g_{q_{-}}$,}\\
|w(t)| &\mbox{otherwise.}
\end{cases}
\end{equation*}
Then we have $\lim_{t\to\pm\infty}v(t)=\pm 1$ and therefore by the usual Modica-Mortola's trick (see for example \cite{modica})
\begin{align*}
\int_{\R}&(|w(t)|^2-1)^2+|w^{\prime}(t)|^2\dt=\int_{\R}(v(t)^2-1)^2+v^{\prime}(t)^2\dt
\\
&\geq \int_{\R}(\tanh(t)^2-1)^2+\tanh^{\prime}(t)^2\dt=\int_{\R}(|w_{q_{-},q_{+}}(t)|^2-1)^2+|w_{q_{-},q_{+}}^{\prime}(t)|^2\dt=\frac{8}{3}
\end{align*}

It therefore only remains to show that $w_{q_{-},q_{+}}\in H^{\times}_{q_{-},q_{+}}(M_k)$. Therefore we choose rotations $R_{q_{-}}$ and $R_{q_{+}}$ such that $R_{q_{-}}e_3=-q_{-}$ and $R_{q_{+}}e_3=q_{+}$ and let $\gamma(t)=\log(\cosh(t))$ be a primitive of $\tanh(t)$. We set
\begin{equation*}
u(t)=
\begin{cases}
R_{q_{-}}(\cos(\gamma(t)+t_0),\sin(\gamma(t)+t_0),0) &\mbox{if $t\leq 0$,}\\
R_{q_{+}}(\cos(\gamma(t)+t_1),\sin(\gamma(t)+t_1),0) &\mbox{if $t>0$,}
\end{cases}
\end{equation*}
where $t_0,t_1$ are chosen such that $R_{q_{-}}(\cos(t_0),\sin(t_0),0)=R_{q_{+}}(\cos(t_1),\sin(t_{1}),0)$ therefore $u$ is continuous at $t=0$. Observing that $\gamma'(0)=0$ we also have that $u'(0)$ exists and is equal to $0$. Then $u\in H^2_{loc}(\R,M_k)$, while a direct computation gives $u\times u^{\prime}=w_{q_{-},q_{+}}$.
\end{proof}

The following compactness result holds true.
\begin{proposition}\label{sbvcompact}
Assume that $\frac{\Ln}{\sqrt{\dn}}\to 0$ and let $z_n=T_{n}(u_{n})$ for some $u_{n}\in C_n(I,M_{k})$ be such that
\begin{equation*}
H^{sl,k}_n(z_n)\leq C\Ln\dn^{\frac{3}{2}}.
\end{equation*}
Then (up to subsequences) $z_n$ converges strongly in $L^1$ to a function $z\in BV(I,Q_k)$.
\end{proposition}

\begin{proof}
By Proposition \ref{weakcompactness} we have that $\sup_n \|z_n\|_{\infty}<+\infty$. Therefore it is enough to show that, up to subsequences, $z_n$ converges in measure to a function $z\in BV(I,Q_k)$. Given $\eta>0$ we define the set
\begin{equation}\label{badset}
A_{\eta}:=\{x\in\rthree:\;\text{dist}(x,L_k)\geq\eta\}.
\end{equation}
We now claim that, for $n$ large enough, we have $z_n^i\notin A_{\eta}$ for all $i\in R_{n}(I)$. Assume by contradiction that the claim does not hold. Passing to a subsequence we have that for each $n$ there exists $i=i(n)$ such that $z_n^{i+1}\in A_{\eta}$. From (\ref{continuity}) we infer that for $n$ large enough 
\begin{equation}\label{assumpt}
z_n^{i+1},\dots,z_n^{i+k}\in A_{\frac{\eta}{2}}.
\end{equation}

As a result we have that for all $j=1,\dots,k$, if $u_n^{i+j}\in S^{1}_{l}$ for some $l\in\{1,2,\dots,k\}$ then $u_n^{i+j+1}\in S^{1}_{m}$ for some $m\neq l$. Moreover, up subsequences we may suppose that, for all $j=1,\dots,k+1$ there exists $l_{j}\in \{1,2,\dots,k\}$ such that 
\begin{equation}\label{11:15}
u_{n}^{i+j}\in S^{1}_{l_{j}},
\end{equation}
where, by the previous discussion we have that $l_{j}\neq l_{j+1}$.
Let $\overline u$ be a limit point for $u_{n}^{i+1}$. Since by Proposition \ref{compactness} $|u_n^{i+j'}-u_n^{i+j}|\to 0$ uniformly in $j,j'\in\{1,\dots,k+1\}$, we have that that for all fixed $j\in\{1,\dots,k+1\}$, $u_{n}^{i+j}\to \overline u$ with the property that $\overline u\in \bigcap_{j=1}^{k+1} S^{1}_{l_{j}}$. Since $|z_n^{i+j}|\geq\frac{\eta}{2}$, by the definition of $z_{n}^{i+j}$ and \eqref{continuity}, for all $j=1,\dots,k+1$ there exists a constant $C=C_{\eta}>0$ such that
\begin{equation}\label{pointbounds}
\frac{1}{C}\dn\leq |u_n^{i+j+1}-u_n^{i+j}|^2\leq C\delta_n\quad\forall j=1\dots,k.
\end{equation}
Thanks to the second inequality above we have 
\begin{equation}\label{angletozero}
\frac{(z_n^{i+j+1},z_n^{i+j})}{|z_n^{i+j+1}||z_n^{i+j}|}\to 1, \quad\forall j=1,\dots,k-1.
\end{equation}
We now claim that 
\begin{equation}\label{boundtransitionpoint}
\frac{1}{C}\delta_n\leq|u_n^{i+j}-\overline u|^2\leq C\delta_n\quad\forall j=1\dots,k+1.
\end{equation}
Indeed, suppose by contradiction that $\frac{|u_n^{i+j}-\overline u|^2}{\dn}\to 0$. Then for $j\in\{1,2,\dots,k\}$ 
\begin{equation*}
z_n^{i+j}=\frac{u_n^{i+j}\times u_n^{i+j+1}}{\sqrt{2\dn}}=\underbrace{\frac{(u_n^{i+j}-\overline u)\times u_n^{i+j+1}}{\sqrt{2\dn}}}_{\to 0}+\underbrace{\frac{\overline u\times u_n^{i+j+1}}{\sqrt{2\dn}}}_{\in L_k}\notin A_{\eta/2},
\end{equation*}
for $n$ large enough, so that the first inequality holds. The case $j=k+1$ is proved with the same argument, exchanging the role of $u_{n}^{i+j}$ and $u_{n}^{i+j+1}$. The second inequality in \eqref{boundtransitionpoint} can be proven as follows. By (\ref{pointbounds}) we have
\begin{align*}
C\geq& \frac{|u_n^{i+j+1}-u_n^{i+j}|^2}{\dn}=\frac{|u_n^{i+j+1}-\overline u|^2}{\dn}+\frac{|u_n^{i+j}-\overline u|^2}{\dn}-\frac{2}{\dn}\frac{(u_n^{i+j+1}-\overline u,u_n^{i+j}-\overline u)}{|u_n^{i+j+1}-\overline u||u_n^{i+j}-\overline u|}|u_n^{i+j+1}-\overline u||u_n^{i+j}-\overline u|
\\
\geq& \left(1-\frac{(u_n^{i+j+1}-\overline u,u_n^{i+j}-\overline u)}{|u_n^{i+j+1}-\overline u||u_n^{i+j}-\overline u|}\right)\left(\frac{|u_n^{i+j+1}-\overline u|^2}{\dn}+\frac{|u_n^{i+j}-\overline u|^2}{\dn}\right)\geq c \frac{|u_n^{i+j}-\overline u|^2}{\dn},
\end{align*} 
where we used that $S^{1}_{l_{j}}\neq S^{1}_{l_{j+1}}$.

By \eqref{pointbounds} and (\ref{boundtransitionpoint}), up to extracting a further subsequence, we have that the sequences $(a_n^j)_{n\in\mathbb{N}}$ definded as
\begin{equation*}
a_n^j:=\frac{u_n^{i+j}-\overline u}{\sqrt{2\dn}}
\end{equation*}
converge to different points $a^j$ for all $j=1,\dots,k+1$. We observe that for all $j=1,\dots,k+1$ we have that $a^{j}$ belongs to the $1$-dimensional subspace $V_{l_{j}}:=\overline u^{\perp}\cap q_{l_{j}}^\perp$ with $l_{j}$ given by \eqref{11:15}. Indeed 
$$
|(a^{j},\overline u)|=\lim_{n}\left|\frac{(u_{n}^{i+j},\overline u)-1}{\sqrt{2\dn}}\right|=\lim_n\frac12\frac{\left|u_{n}^{i+j}-\overline u\right|^2}{\sqrt{2\dn}}=0
$$
by \eqref{boundtransitionpoint}. On the other hand $(a^{j},q_{l_{j}})=0$ simply follows by \eqref{11:15} since $q_{l_{j}}\perp S^{1}_{l_{j}}$.
%It is easy to see that $a^j\in\mathbb{R}^2\times\{0\}$. 
We now show that all $a^j$ are collinear. Indeed, by definition \eqref{mapT} we have that 
\begin{align*}
|z^{i+j}_{n}|=\frac{\sqrt{1-(u_{n}^{i+j+1},u_{n}^{i+j})^{2}}}{\sqrt{2\dn}}.
\end{align*}
On the other hand, again by \eqref{mapT} and the well-known formula $(a\times b,c\times d)=(a,c)(b,d)-(b,c)(a,d)$ we have that
\begin{align*}
(u_n^{i+j+2}-u_n^{i+j+1},u_n^{i+j+1}-u_n^{i+j})-2\dn(z_n^{i+j+1},z_n^{i+j})\\=-(1-(u_{n}^{i+j+2},u_{n}^{i+j+1}))(1-(u_{n}^{i+j+1},u_{n}^{i+j}))
\end{align*}
From the previous two equalities, together with Proposition \ref{compactness} and \eqref{angletozero} we then get that for all $j\in\{1,2,\dots,k-1\}$
\begin{align*}
1\geq&\frac{(a^{j+2}-a^{j+1},a^{j+1}-a^j)}{|a^{j+2}-a^{j+1}||a^{j+1}-a^j|}=\lim_{n\to\infty}\frac{(u_n^{i+j+2}-u_n^{i+j+1},u_n^{i+j+1}-u_n^{i+j})}{\sqrt{(2-2(u_n^{i+j+2},u_n^{i+j+1}))(2-2(u_n^{i+j+1},u_n^{ji+}))}}
\\
=&\lim_{n\to\infty}\frac{2\dn(z_n^{i+j+1},z_n^{i+j})+(u_n^{i+j+2}-u_n^{i+j+1},u_n^{i+j+1}-u_n^{i+j})-2\dn(z_n^{i+j+1},z_n^{i+j})}{\sqrt{(2-2(u_n^{i+j+2},u_n^{i+j+1}))(2-2(u_n^{i+j+1},u_n^{i+j}))}}
\\
\geq&\lim_{n\to\infty}\frac{(z_n^{i+j+1},z_n^{i+j})}{|z_n^{i+j+1}||z_n^{i+j}|}\frac{1}{2}\sqrt{(1+(u_n^{i+j+2},u_n^{i+j+1}))(1+(u_n^{i+j+1},u_n^{i+j}))}
\\
&-\lim_{n\to\infty}\frac{1}{2}\sqrt{((1-(u_n^{i+j+2},u_n^{i+j+1}))(1-(u_n^{i+j+1},u_n^{i+j}))}=1,
\end{align*}
which is equivalent to say that all $a^{j}$ with $j\in\{1,2,\dots,k+1\}$ are collinear. This gives a contradiction, as the line containing the $a^{j}$'s should then intersect in $k+1$ distinct points the set $\bigcup_{j=1}^{k}V_{l_{j}}$, which instead consists of at most $k$ $1$-dimensional linear subspaces. This proves our claim that for $n$ large enough $z^{i}_{n}\not\in A_{\eta}$ for all $i\in R_{n}(I)$ which implies that
${\rm dist}(z_{n},L_{k})\to 0\quad \text{uniformly}$. Since as in the proof of Proposition \ref{weakcompactness} $|z_{n}|\to 1$ almost everywhere in $I$ we deduce that 
\begin{equation}\label{measconv}
{\rm dist}(z_{n},Q_{k})\to 0\quad \text{in measure}.
\end{equation}

Now we argue similar to the proof of Lemma 6.2 in \cite{GCB}. At first we chose $r>0$ such that the family of balls $\{B_{3r}(z)\}_{z\in Q_k}$ is pairwise disjoint. We set
\begin{equation*}
d:=\inf_{\substack{q_1,q_2\in Q_k\\q_1\neq q_2}}\text{dist}(B_{2r}(q_1),B_{2r}(q_2))>0.
\end{equation*}
Suppose $z_n$ takes values in different balls $B_r(q_1)$ and $B_r(q_2)$. Then, by \eqref{continuity} there exists a path $z_n^i,\dots,z_n^{i+j}$ such that $r<|z_n^i-q_1|<2r$ and $r<|z_n^{i+j}-q_2|<2r$ and such that 
\begin{equation*}
z_n^l\notin\bigcup_{q\in Q_k}B_{2r}(q)\quad\forall i<l<i+j.
\end{equation*} 
Defining $A_{\eta}$ as in (\ref{badset}), from the first part of the proof we know that $z_n\notin A_{\eta}$ for $n$ large enough. Choosing a suitable $\eta=\eta(r)$ we deduce that, for $n$ large enough,
 \begin{equation}\label{uandz}
\inf_{l=i,\dots,i+j}\text{dist}(z_n^l,S^2)\geq r,\quad |z_n^i-z_n^{i+j}|\geq d.
\end{equation}
Now we use the classical Modica-Mortola trick to estimate the energy of such a path. By the uniform energy bound $H_{n}^{sl,k}(z_{n})\leq C\Ln\dn^{3/2}$ we get  $H_{n}^{sl}(u_{n})\leq C\Ln\dn^{3/2}$. 
Since $z_n$ is uniformly bounded by Proposition \ref{weakcompactness} and $(u_n^{i+1},u_n^i)$ converges uniformly to $1$ by Proposition \ref{compactness} we may then write, for $n$ large enough, the following estimate
\begin{equation}\label{uzclose}
|1-(u_n^{i+1},u_n^i)|\leq 3(1-(u_n^{i+1},u_n^i)^2)\leq 6\dn|z_n^i|^2\leq C\dn.
\end{equation}
As a result we have
\begin{equation*}
\left|\left|\frac{u_n^{i+1}-u_n^i}{\sqrt{2\dn}}\right|^2-|z_n^i|^2\right|= \frac{(1-(u_n^{i+1},u_n^i))^2}{2\dn}\leq C\dn.
\end{equation*} 
so that using \eqref{uandz} it holds that
\begin{align}\label{mmtrick}
d\leq& \sum_{l=i}^{i+j-1}|z_n^{l+1}-z_n^l|\leq \frac{2}{r}\sum_{l=i}^{i+j-1}\Ln \left|\left|\frac{u_n^{i+1}-u_n^i}{\sqrt{2\dn}}\right|^2-1\right|\left|\frac{z_n^{l+1}-z_n^l}{\Ln}\right|\nonumber
\\
\leq &\frac{2}{r}\frac{\sqrt{\dn}}{\Ln}\sum_{l=i}^{i+j-1}\Ln\left(\left|\frac{u_n^{i+1}-u_n^i}{\sqrt{2\dn}}\right|^2-1\right)^2+\frac{2}{r}\frac{\Ln}{\sqrt{\dn}}\sum_{l=i}^{i+j-1}\Ln\left|\frac{z_n^{l+1}-z_n^l}{\Ln}\right|^2,
\end{align}
from which we deduce that such a transition costs a finite amount of positive energy, depending only on $r$. Thus we have only finitely many of these transitions, their number being bounded uniformly with respect to $n$. It follows that, up to subsequences, $z_n$ converges in measure to a piecewise constant function with values in $Q_k$. We omit the details.

\end{proof}

After establishing compactness for sequences with equi-bounded energy, we are in a position to prove the following $\Gamma$-convergence result. In the proof of the lower bound we will make use of the area formula for absolutely continuous function, which we briefly recall: for every positive Borel function $h$, every absolutely continuous function $\zeta: [a,b]\to \R$ it holds
\begin{equation}\label{area}
\int_{\zeta([a,b])} \left(\sum_{s\in\zeta^{-1}(v)}h(s)\right)\,\mathrm{d}v=\int_a^b h(s)|\zeta'(s)|\,\mathrm{d}s
\end{equation}
(see \cite[Theorem 3.65]{leoni-sobolev}).

\begin{theorem}\label{main3}
Let $H_n^{sl,k}:L^1(I)\rightarrow [0,+\infty]$ be defined as in (\ref{newfunctional}). Assume that $\frac{\Ln}{\sqrt{\dn}}\to 0$. Then the functionals $\frac{H_n^{sl,k}}{\sqrt{2}\Ln\dn^{\frac{3}{2}}}$ $\Gamma$-converge with respect to the strong $L^1$-topology to the functional 
\begin{equation*}
H^{sl,k}(z)=
\begin{cases}
\frac{8}{3}\#S(z) &\mbox{if $z\in BV(I,Q_k)$,}\\
+\infty &\mbox{otherwise.}
\end{cases}.
\end{equation*}
\end{theorem}

\begin{proof}
We start with the lower bound. Without loss of generality we may consider $z_{n}=T_{n}(u_{n})$ for some $u_n\in C_n(I,M_{k})$ such that $z_n\to z$ in $L^1(I,\R^3)$ and 
\begin{equation*}
\liminf_n\frac{H_n^{sl}(z_n)}{\sqrt{2}\Ln\dn^{\frac{3}{2}}}\leq C<+\infty.
\end{equation*}
From Proposition \ref{sbvcompact} we know that $z\in BV(I,Q_k)$. Furthermore, if we denote by  $z_n^a$ denote the piecewise affine interpolation of $z_n$ on the lattice $\Ln\mathbb{Z}\cap I$, we also have that $z^a_n\to z$ in $L^1(I,\R^3)$. Passing to a subsequence (not relabeled) we can assume that $z^a_n$ converges to $z$ almost everywhere.
Furthermore, for all $\eta >0$, defining $A_{\eta}$ as in (\ref{badset}), for $n$ large enough, we have $z^i_n\notin A_{\eta}$ for all $i \in R_n(I)$ and this in turn implies that
\begin{equation}\label{uniform}
\mbox{dist}(z^a_n, L_k) \to 0
\end{equation}
uniformly. Let now $t_1<\dots<t_l$ be the jump set of $z$. Let $\alpha>0$ be such that $[-2\alpha+t_m,t_m+2\alpha]\cap[-2\alpha+t_j,t_j+2\alpha]=\emptyset$ for  all $j\neq m$. By the choice of $\alpha$ it holds that
\begin{equation*}
\liminf_{n}\frac{H_n^{sl,k}(z_n)}{\sqrt{2}\Ln\dn^{\frac{3}{2}}}\geq\sum_{m=1}^l\liminf_n F^m_n(z_n),
\end{equation*} 
where 
\begin{equation*}
F^m_n(z_n)=\frac{\sqrt{2\dn}}{\Ln}\sum_{|\Ln i-t_m|<2\alpha}\Ln \left(\left|\frac{u_n^{i+1}-u_n^i}{\sqrt{2\dn}}\right|^2-1\right)^2
+\frac{\Ln}{\sqrt{2\dn}}\sum_{|\Ln i-t_m|<2\alpha}\Ln\left|\frac{z_n^{i+1}-z_n^i}{\Ln}\right|^2.
\end{equation*}
We now fix $t_m$ and to reduce notation we set $q_{\pm}:=z(t_m\pm\alpha)$. Our goal is to show that
\begin{equation*}
\liminf_n F_n^m(z_n)\geq\frac{8}{3}\,,
\end{equation*}
which yields the lower bound.

To prove our claim, we begin by observing that, due to almost everywhere convergence, we can assume that
\begin{equation}\label{qconvergence}
z^a_n(t_m\pm\alpha)\to q_{\pm}
\end{equation}
when $n\to +\infty$. 
Furthermore, since  
\begin{equation}\label{uniformclose}
\left|\frac{u_n^{i+1}-u_n^i}{\sqrt{2\dn}}\right|^2=\frac{2}{1+(u_n^{i+1},u_n^i)}|z_n^i|^2
\end{equation}
we can write
\begin{equation*}
\frac{\sqrt{2\dn}}{\Ln}\sum_{|\Ln i-t_m|<2\alpha}\Ln \left(\left|\frac{u_n^{i+1}-u_n^i}{\sqrt{2\dn}}\right|^2-1\right)^2=\frac{\sqrt{2\dn}}{\Ln}\sum_{|\Ln i-t_m|<2\alpha}\Ln \left(|\beta_n^i z_n^i|^2-1\right)^2
\end{equation*}
where we have denoted by $\beta_n\in C_n(I,\R)$ the sequence of piecewise constant functions such that $\beta_{n}^{i}=\frac{2}{1+(u_n^{i+1},u_n^i)}$ which converges uniformly to $1$. We now show that we can switch from the piecewise constant interpolation to the affine one without increasing the energy. Indeed, given $\sigma>0$, we have
\begin{align*}
\int_{-\alpha+t_m}^{\alpha+t_m}&\left(|\beta_n(s)z_n^a(s)|^2-1\right)^2\,\mathrm{d}s\leq (1+\sigma)\int_{-\alpha+t_m}^{\alpha+t_m}\left(|\beta_n(s)z_n(s)|^2-1\right)^2\,\mathrm{d}s
\\
&\hspace*{4cm}+\left(1+\frac{1}{\sigma}\right)\int_{-\alpha+t_m}^{\alpha+t_m}|\beta_n(s)|^4\left(|z_n^a(s)|^2-|z_n(s)|^2\right)^2\,\mathrm{d}s
\\
\leq&(1+\sigma)\sum_{|\Ln i-t_m|<2\alpha}\Ln \left(|\beta_n^i z_n^i|^2-1\right)^2+C\left(1+\frac{1}{\sigma}\right)\int_{-\alpha+t_m}^{\alpha+t_m}|z_n^a(s)-z_n(s)|^2\,\mathrm{d}s
\\
\leq&(1+\sigma)\sum_{|\Ln i-t_m|<2\alpha}\Ln \left(|\beta_n^i z_n^i|^2-1\right)^2+C\left(1+\frac{1}{\sigma}\right)\sum_{|\Ln i-t_m|<2\alpha}\Ln|z_n^{i+1}-z_n^i|^2
\\
\leq&(1+\sigma)\sum_{|\Ln i-t_m|<2\alpha}\Ln \left(|\beta_n^i z_n^i|^2-1\right)^2+C\left(1+\frac{1}{\sigma}\right)\sqrt{\dn}\Ln,
\end{align*}
where we have used the energy bound \eqref{liminfbound}  and the fact that both $\beta_n$ and $z_n$ are equibounded sequences. Multiplying the last inequality by $\frac{\sqrt{2\dn}}{\Ln}$ we obtain
\begin{equation*}
\liminf_n\frac{\sqrt{2\dn}}{\Ln}\int_{-\alpha+t_m}^{\alpha+t_m}\!\!\!\left(|\beta_n(s)z_n^a(s)|^2-1\right)^2\mathrm{d}s\leq (1+\sigma)\liminf_n\frac{\sqrt{2\dn}}{\Ln}\!\!\!\!\sum_{|\Ln i-t_m|<2\alpha}\!\!\!\!\Ln \left(|\beta_n^i z_n^i|^2-1\right)^2. 
\end{equation*}
By the arbitrariness of $\sigma$ we deduce that 
\begin{align}
\liminf_{n} F^m_n(z_n)\geq&\liminf_n\frac{\sqrt{2\dn}}{\Ln}\int_{-\alpha+t_m}^{\alpha+t_m}\left(|\beta_n(s)z_n^a(s)|^2-1\right)^2\,\mathrm{d}s\nonumber
\\
&+\liminf_n\frac{\Ln}{\sqrt{2\dn}}\int_{-\alpha+t_m}^{\alpha+t_m}|(z_n^a)^{\prime}(s)|^2\,\mathrm{d}s\label{tocontinuum}.
\end{align}
We now fix an arbitrary $\varepsilon >0$: due to \eqref{qconvergence}, when $n$ is sufficiently large we have
\begin{equation}\label{moduloquasi1}
\left|z^a_n(t_m\pm\alpha)\right|\ge \frac{1}{1+\varepsilon}\,. 
\end{equation}
Furthermore, using \eqref{uniform}, the continuity of $z_n^a$ and \eqref{qconvergence}, for all $n$ sufficiently large there exists a point $\tau_n \in~(t_m-\alpha, t_m+\alpha)$ such that
\begin{equation}\label{moduloquasi0}
\left|z^a_n(\tau_n)\right|\le \varepsilon\,. 
\end{equation}

We define the absolutely continuous function $\zeta_n$ by  $\zeta_n(s):=|z^a_n(s)|$. Applying the Cauchy-Schwarz inequality to the right-hand side of \eqref{tocontinuum}, and taking into account that $|\zeta_n^{\prime}|\le |(z_n^a)^{\prime}|$ we have
\begin{align}
\liminf_nF^m_n(z_n)\geq& \liminf_n2\int_{-\alpha+t_m}^{\alpha+t_m}\left||\beta_n(s)z_n^a(s)|^2-1\right|\,\left|(z_n^a)^{\prime}(s)\right|\,\mathrm{d}s\nonumber
\\
\geq& \liminf_n2\int_{-\alpha+t_m}^{\alpha+t_m}\left|\beta_n(s)^2\,\zeta_n(s)^2-1\right|\,\left|\zeta_n^{\prime}(s)\right|\,\mathrm{d}s\label{mod-mor}\\
\geq&\liminf_n2\int_{-\alpha+t_m}^{\tau_n}\left|\beta_n(s)^2\,\zeta_n(s)^2-1\right|\,\left|\zeta_n^{\prime}(s)\right|\,\mathrm{d}s
\\&+\liminf_n2\int_{\tau_n}^{\alpha+t_m}\left|\beta_n(s)^2\,\zeta_n(s)^2-1\right|\,\left|\zeta_n^{\prime}(s)\right|\,\mathrm{d}s\,.\nonumber
\end{align}
Using formula \eqref{area} with $h(s)=\left|\beta_n(s)^2\,\zeta_n(s)^2-1\right|$ and $\zeta=\zeta_n$ and observing that, by \eqref{moduloquasi1} and \eqref{moduloquasi0}, $[\varepsilon, \frac{1}{1+\varepsilon}]\subseteq \zeta_n([-\alpha+t_m, \tau_n])$, we have
$$
\int_{-\alpha+t_m}^{\tau_n}\left|\beta_n(s)^2\,\zeta_n(s)^2-1\right|\,\left|\zeta_n^{\prime}(s)\right|\,\mathrm{d}s\geq\int_{\varepsilon}^{\frac{1}{1+\varepsilon}} \left(\sum_{s\in\zeta_n^{-1}(v)}|\beta_n(s)^2\,v^2-1|\right)\,\mathrm{d}v\,.
$$
Since $\beta_n \to 1$ uniformly, when $n$ is large enough we have that $\beta_n(s)\le 1+\e$ for all $s$. Using the elementary inequality
$$
|\theta ^2 v^2-1|=1-\theta ^2 v^2\geq 1-(1+\e)^2 v^2
$$
for all $\theta \in [0, 1+\e]$ and $v \in [\varepsilon, \frac{1}{1+\varepsilon}]$, we deduce that
$$
\int_{-\alpha+t_m}^{\tau_n}\left|\beta_n(s)^2\,\zeta_n(s)^2-1\right|\,\left|\zeta_n^{\prime}(s)\right|\,\mathrm{d}s\geq\int_{\varepsilon}^{\frac{1}{1+\varepsilon}}(1-(1+\e)^2 v^2)\,\mathrm{d}v\,.
$$
The same estimate holds also for the other summand in the right-hand side of \eqref{mod-mor}. Therefore we conclude
$$
\liminf_n F^m_n(z_n)\geq 4\int_{\varepsilon}^{\frac{1}{1+\varepsilon}}(1-(1+\e)^2 v^2)\,\mathrm{d}v\,.
$$
Since $\e$ was arbitrary, we conclude that
$$
\liminf_n F^m_n(z_n)\geq 4\int_{0}^{1}(1- v^2)\,\mathrm{d}v=\frac83\,,
$$
which gives the required lower bound.

\medskip

The upper bound follows as in the proof of Theorem \ref{main1}. We only indicate here the major changes. Since the argument is local, let us assume that $z=q_1\mathds{1}_{[0,\frac{1}{2})}+q_2\mathds{1}_{(\frac{1}{2},1]}$ for some $q_{1},q_{2}\in Q_{k}$. Given $\e>0$ we set
\begin{equation*}
w_{\e}(t)=
\begin{cases}
|f_{\e}(t)|q_1 &\mbox{if $t\leq 0$,}\\
|f_{\e}(t)|q_2 &\mbox{if $t>0$,}
\end{cases}
\end{equation*}
where $f_{\e}$ is defined by the construction below. Let $t_{\e}>0$ be such that $|\tanh(\pm t_{\e})-(\pm 1)|\leq \e$ and
\begin{equation*}
\int_{|t|\leq t_{\e}}\left(|\tanh(t)|^2-1\right)^2+|\tanh^{\prime}(t)|^2\dt\geq \frac{8}{3}-\e.
\end{equation*}
We then define $f_{\e}:\R\to\R$ as an odd $C^{1}$ function such that 
\begin{equation}\label{optvelocity}
f_{\e}(t):=\begin{cases}
\tanh(t)&\mbox{if $t\in[0,t_{\e}]$,}\\
p_{\e}(t)&\mbox{if $t\in(t_{\e},t_{\e}+\e)$,}\\
1&\mbox{if $t\in(t_{\e}+\e,+\infty)$,}
\end{cases}
\end{equation}
where $p_{\e}$ is a suitable third order interpolating polynomial that we may choose such that $\|p_{\e}'\|_{\infty}\leq 2$. Note that $w_{\e}\in H^{\times}_{q_1,q_2}(M_k)$ and by construction
\begin{equation*}
\int_{\R}\left(|w_{\e}(t)|^2-1\right)^2+|w_{\e}^{\prime}(t)|^2\dt\leq \frac{8}{3}+C\e
\end{equation*}
for some constant $C>0$. Let $u\in H^2_{loc}(\R,M_k)$ be such that $w_{\e}=u\times u^{\prime}$. For each $n\in\mathbb{N}$ we let
\begin{equation*}
\alpha_n=\frac{\arccos(1-\dn)}{\sqrt{2\dn}}.
\end{equation*}
and then define the function $u_n\in C_n(I,S^2)$ setting
\begin{equation*}
u_n^i=u\left(\alpha_n\frac{\sqrt{2\dn}}{\Ln}(\Ln i-\frac{1}{2})\right).
\end{equation*} 
From now on we proceed as in the proof of Theorem \ref{main1}, the only change is that the corresponding function $u$ is not twice differentiable in the origin. But this does not really affect the argument. We obtain
\begin{equation*}
\Gamma-\limsup_n \frac{H_n^{sl,k}(z)}{\sqrt{2}\Ln\dn^{\frac{3}{2}}}\leq\frac{8}{3}+C\e,
\end{equation*}
which yields the claim by the arbitrariness of $\e$.
\end{proof}

\subsubsection{Chirality transitions via {\it soft} penalization}\label{anglesect}

In the previous model we forced the spin variable to take values only in finitely many rotated copies of $S^1$. As we have seen, this restriction leads to a positive limit energy when changing the chirality. However, this energy is independent of the distance between two chirality vectors in contrast to the results conjectured in \cite{DmiKriext}. To obtain such a dependence we propose another model, where we penalize the distance of $u$ from the set $M_k$ with an additional energy term. Choosing the right scaling this penalization preserves compactness, but yields more freedom for the optimal chirality transition. Given $u\in C_n(I,S^2)$ we define the already normalized new energy by
\begin{equation}\label{penalenergy}
H^{p}_n(u)=H^{sl}_n(u)+\mu_n\sum_{i\in R_n(I)}\Ln G(u^i\times u^{i+1}),
\end{equation} 
where $\mu_n>0$ and $G:\R^3\backslash\{0\}\rightarrow [0,+\infty)$ is a continuous, zero-homogeneous function that we consider extended at $0$ setting $G(0):=0$ and such that
\begin{equation}\label{penalcontrol}
\{z\in\R^3:\;G(z)=0\}=L_k,
\end{equation}
with $L_{k}$ as in \eqref{Hpen}.
Without changing notation we define $H^p_n: L^1(I,\R^3)\rightarrow [0,+\infty)$ in the $z$-variable setting
\begin{equation}\label{penalzvariable}
H_n^{p}(z)=
\begin{cases}
\inf_{T(u)=z}H^{p}_n(u) &\mbox{if $z=T(u)\text{ for some }u\in C_n(I,S^2)$,}
\\
+\infty &\mbox{otherwise.}
\end{cases}
\end{equation}
For $Q_{k}$ as in \eqref{Hpen}, we introduce $h_G:Q_k\times Q_k\rightarrow [0,+\infty)$ setting
\begin{equation}\label{penaltransition}
h_G(q_1,q_2):=\inf\left\{\int_{\R}(|w(t)|^2-1)^2+\frac{G(w(t))}{2}\,\mathrm{d}t+\int_{\R}|w^{\prime}(t)|^2\,\mathrm{d}t:\;w\in H^{\times}_{q_1,q_2}\right\}.
\end{equation}
Note that $h_G(q_1,q_2)\leq\frac{8}{3}$ since the minimizer of the optimal profile problem defined in (\ref{newtransitionenergy}) is admissible and $G$ vanishes by (\ref{penalcontrol}). 

For the penalized energies the following compactness result holds true.
\begin{lemma}\label{penalcompact}
Assume that $\lim_n\frac{\Ln}{\sqrt{\dn}}= 0$ and $\liminf_n\frac{\mu_n}{\dn^2}\geq c_{\mu}>0$. Let $z_n\in C_n(I,\R^3)$ be such that
\begin{equation*}
H^{p}_n(z_n)\leq C\Ln\dn^{\frac{3}{2}}.
\end{equation*}
Then (up to subsequences) $z_n$ converges strongly in $L^1$ to a function $z\in BV(I,Q_k)$.
\end{lemma}

\begin{proof}
Applying Lemma \ref{weakcompactness} we infer that $\|z_n\|_{\infty}$ is uniformly bounded so it is enough to prove convergence in measure. Without loss of generality we assume that $\frac{\mu_n}{\sqrt{2}\Ln\dn^{\frac{3}{2}}}\geq \frac{\sqrt{2\dn}}{\Ln}$. Then, by (\ref{liminfbound}), we may find a vanishing sequence $\gamma_n>0$ such that

\begin{equation*}
C\geq\frac{\sqrt{2\dn}}{\Ln}\sum_{i\in R_n(I)}\Ln\left[\left(\left|\frac{u_n^{i+1}-u_n^i}{\sqrt{2\dn}}\right|^2-1\right)^2+\frac{G(z_n^i)}{2}\right]+\frac{(1-\gamma_n)\Ln}{\sqrt{2\dn}}\sum_{i\in R_n(I)}\Ln\left|\frac{z_n^{i+1}-z_n^i}{\Ln}\right|^2.
\end{equation*}
Defining $W(z)=(|z|^2-1)^2+\frac{G(z)}{2}$ we have that $W$ is non-negative, lower semicontinuous with zeros exactly in $S^2\cap L_k=Q_k$. Therefore, if we consider the set $Q_k^{\eta}:=\{z\in\R^3:\;\text{dist}(z,Q_k)\geq\eta\}$, by a coercivity argument we have 
\begin{equation}\label{wellgrowth}
\inf_{z\in Q_k^{\eta}}W(z)=\min_{z\in Q_k^{\eta}}W(z)=c_{\eta}>0.
\end{equation} 
Combining (\ref{wellgrowth}) with (\ref{uzclose}) we deduce that $z_n$ converges in measure to the set $Q_k$. The rest of the statement follows now arguing as in the proof of Proposition \ref{sbvcompact}.
\end{proof}

Before we prove a $\Gamma$-convergence result, we need the following two auxiliary lemmata. Roughly speaking, the first one states that we can connect two paths, that are near to the same point in $Q_k$, by paying very small energy.

\begin{lemma}\label{connectpoints}
Let $0<\eta<<1$ be small and let $w_0,w_1\in \R^3$ be such that there exists $\hat{q}\in Q_k$ with $\max_i|w_i-\hat{q}|\leq\eta$. Moreover, for $i=0,1$, let $u_{i}\in S^2\cap w_{i}^{\perp}$. Then there exists an interval $[0,t^*]\subset[0,3+4\pi]$ and a $C^2$-function $u:[0,t^*]\rightarrow S^2$ such that, setting $w=u\times u^{\prime}$, it holds 
\begin{align*}
&u(0)=u_0,\,u(t^*)=u_1,
\\
&w(0)=w_0,\,w(t^*)=w_1,
\\
&\int_0^{t^*}\left(|w(t)|^2-1\right)^2+\frac{G(w(t))}{2}\dt+\int_0^{t^*}|w^{\prime}(t)|^2\dt\leq C_{\eta},
\end{align*}
with $\lim_{\eta\to 0}C_{\eta}=0$.
\end{lemma}

\begin{proof}
First note that if $\eta$ is small enough, for $i=0,1$ we have $w_{i}\neq 0$ and $|\frac{w_{i}}{|w_{i}|}-\hat{q}|\leq 2\eta$. Thus, for every $z\in\{ sw_{i}+(1-s)\frac{w_{i}}{|w_{i}|}:\;s\in[0,1]\}$ we have the estimate $G(z)\leq c_G(2\eta)$, where $c_G$ is a modulus of continuity of $G_{|S^2}$. Moreover let $R_0\in SO(3)$ be such that $\frac{w_0}{|w_0|}=R_0e_3$ and $R_0^T\frac{w_1}{|w_1|}=\lambda e_2+\mu e_3$.
\\
\hspace*{0,5cm}
We start constructing a path joining $w_{0}$ and $\frac{w_0}{|w_0|}$. Let us choose a $C^2$-function $\gamma_0:[0,1]\rightarrow \R$ with the following properties:
\begin{itemize}
\item[(i)] $\gamma_0(0)=\gamma_0(1)=0,$
\item [(ii)] $\gamma_0^{\prime}(0)=|w_0|-1,\,\gamma_0^{\prime}(1)=0,$
\item [(iii)] $\gamma_0^{\prime\prime}(0)=\gamma_0^{\prime\prime}(1)=0.$
\end{itemize}
Since $||w_0|-1|\leq \eta$ we can choose the function $\gamma_0$ such that 
\begin{equation}\label{smallinterpol}
\max\{\|\gamma_0^{\prime}\|_{\infty},\|\gamma_0^{\prime\prime}|_{\infty}\}\leq C\eta.
\end{equation}
Now we define $u_0:[0,1]\rightarrow S^2$ via
\begin{equation*}
u_0(t)=R_0(\cos(t+\gamma_0(t)+t_0),\sin(t+\gamma_0(t)+t_0),0),
\end{equation*}
where $t_0\in[0,2\pi)$ is such that $R_0(\cos(t_0),\sin(t_0),0)=u_0$. We further set $w_0=u_0\times u_0^{\prime}$. Then we have $u_0(0)=u_0$, $w_0(t)=(1+\gamma_0^{\prime}(t))\frac{w_0}{|w_0|}$ and $w_0^{\prime}(t)=\gamma_0^{\prime\prime}(t)\frac{w_0}{|w_0|}$, and therefore $w_0(0)=w_0,\,w_0(1)=\frac{w_0}{|w_0|}$, and
\begin{equation*}
\int_0^{1}\left(|w_0(t)|^2-1\right)^2+\frac{G(w_0(t))}{2}\dt+\int_0^1|w_0^{\prime}(t)|^2\dt\leq C\eta^2+c_G(2\eta).
\end{equation*}
We continue by joining $\frac{w_0}{|w_0|}$ and $\frac{w_1}{|w_1|}$. Let us take $B$ as a suitable logarithm of the matrix
$$
A=\begin{pmatrix}
1 &0 &0 \\
0 &\mu &\lambda \\
0 &-\lambda &\mu
\end{pmatrix}.$$
Now we choose a $C^2$ cut-off function $\gamma_1:\R\rightarrow [0,1]$ such that
\begin{equation*}
\gamma_1(t)=
\begin{cases}
0 &\mbox{$t\leq 1$,}\\
1 &\mbox{$t\geq 2$,}
\end{cases}
\end{equation*}
We set $u_1:[1,2]\rightarrow S^2$ as
\begin{equation*}
u_1(t)=R_0\exp(\gamma(t)B)(\cos(t+t_0),\sin(t+t_0),0).
\end{equation*}
Defining $w_1:[1,2]\rightarrow\R^3$ via $w_1=u_1\times u_1^{\prime}$, by the same calculations as in the proof of Lemma \ref{bvbound} we get
\begin{align}
&\left(|w_1(t)|^2-1\right)^2\leq C\left(|B|+|B|^2\right)^2,\label{doublewell}
\\
&|w_1^{\prime}(t)|^2\leq C\left(|B|+|B|^2\right)^2.\label{deriv}
\end{align}

In order to estimate $G(w_1(t))$, observe that for one particular matrix logarithm and the Frobenius norm, we have $\|B\|_F=|\arccos(\frac{\text{tr}(A)-1}{2})|=|\arccos(\mu)|$, so that by the equivalence of all matrix norms we infer $|B|\leq C|\arccos(\mu)|$. Moreover, it holds that
\begin{equation*}
4\eta\geq |\frac{w_0}{|w_0|}-\frac{w_1}{|w_1|}|=|e_3-\lambda e_2-\mu e_3|\geq |1-\mu|,
\end{equation*}
so if $\eta$ is small enough, we have $|B|<1$. We deduce that
\begin{align}
|w_1(t)-\frac{w_0}{|w_0|}|&=|u_1(t)\times u_1^{\prime}(t)-\frac{w_0}{|w_0|}|\leq C|B|+|\exp(\gamma(t)B)e_3-e_3|\nonumber
\\
&\leq C|B|+\exp(|B|)-1\leq (C+\exp(1))|B|.\label{estpenal}
\end{align}
By calculating the leading order term of $\arccos^2(x)$ at $x=1$ we get
\begin{equation*}
|B|^2\leq C|\arccos(\mu)^2|\leq C|1-\mu|\leq C\,\eta,
\end{equation*}
so that, combined with (\ref{estpenal}) we have
\begin{equation}\label{Gestimate}
|\frac{w_1(t)}{|w_1(t)|}-\frac{w_0}{|w_0|}|\leq C\,\sqrt{\eta}+\left|1-|w_1(t)|\right|\leq C\sqrt{\eta},
\end{equation}
which implies $G(w_1(t))\leq c_G(C\sqrt{\eta})$. Integrating (\ref{doublewell}), (\ref{deriv}) and the previous bound we infer
\begin{equation*}
\int_1^{2}\left(|w_1(t)|^2-1\right)^2+\frac{G(w_1(t))}{2}\dt+\int_1^2|w_1^{\prime}(t)|^2\dt\leq C\,\eta+c_G(C\sqrt{\eta}).
\end{equation*}

As a last part we join $\frac{w_1}{|w_1|}$ and $w_1$. We define $\gamma_2:[2,+\infty)\rightarrow\R$ setting
\begin{equation*}
\gamma_2(t)=
\begin{cases}
\tilde{\gamma}_0(3-t) &\mbox{if $t\leq 3$,}\\
(|w_1|-1)(t-3) &\mbox{otherwise,}
\end{cases}
\end{equation*}
where $\tilde{\gamma}_0$ fulfills the same requirements as $\gamma_0$ with $w_1$ instead of $w_0$. Then $\gamma_2$ is of class $C^2$ and defining $u_2(t)=R_0A(\cos(t+\gamma_2(t)+t_0),\sin(t+\gamma_2(t)+t_0),0)$ as well as $w_2=u_2\times u_2^{\prime}$ we have
\begin{equation*}
\int_2^{3}\left(|w_2(t)|^2-1\right)^2+\frac{G(w_2(t))}{2}\dt+\int_2^3|w_2^{\prime}(t)|^2\dt\leq C\eta^2+c_G(2\eta).
\end{equation*}
By the intermediate value theorem, if $\eta\leq\frac{1}{2}$ there exists $t^*\in [3,3+4\pi]$ such that $u_2(t^*)=u_1$, where we have used that $R_0Ae_3=\frac{w_1}{|w_1|}$, so that $u_2(t)\in w_1^{\perp}$. Moreover, since $1+\gamma_2^{\prime}(t)=|w_1|$ for $t\geq 3$, one can easily show that
\begin{equation*}
\int_3^{t^*}\left(|w_2(t)|^2-1\right)^2+\frac{G(w_2(t))}{2}\dt+\int_3^{t^*}|w_2^{\prime}(t)|^2\dt\leq C\eta^2+c_G(2\eta).
\end{equation*}
Finally we set $J=[0,t^*]$ and a lenghty, but straightforward calculation shows that if we define $u:[0,t^*]\rightarrow S^2$ as
\begin{equation*}
u(t)=
\begin{cases}
u_0(t) &\mbox{if $t\in [0,1]$,}\\
u_1(t) &\mbox{if $t\in[1,2]$,}\\
u_2(t) &\mbox{if $t\in[2,t^*]$,}
\end{cases}
\end{equation*} 
we preserve the $C^2$-regularity. By construction this function fulfills all required properties since $G_{|S^2}$ is uniformly continuous.
\end{proof}

For technical reasons we need to show that the class of admissible functions defining $h_G(q_1,q_2)$ in (\ref{penaltransition}) can be taken to be more regular.

\begin{lemma}\label{regularclass}
Let $q_{\pm}\in Q_k$. Then the infimum in (\ref{penaltransition}) can be taken equivalently over all functions $w\in W_{loc}^{2,\infty}(\R,\R^3)\cap H^{\times}_{q_-,q_+}$ such that 
\begin{align*}
w(t)&=q_-\quad\forall t\leq t_1,
\\
w(t)&=q_+\quad\forall t\geq t_2
\end{align*}
for some $t_1<t_2$, $u^{\prime\prime}$ is piecewise continuous and the set $\{w(t)=0\}$ is finite.
\end{lemma}

\begin{proof}
Given $\e>0$ we find a function $\tilde{u}\in H^2_{loc}(\R,S^2)$ such that $\tilde{w}=\tu\times \tu^{\prime}$ is admissible in the infimum problem defining $h_G(q_-,q_+)$ in (\ref{penaltransition}) and
\begin{equation*}
\int_{\R}(|\tilde{w}(t)|^2-1)^2+\frac{G(\tilde{w}(t))}{2}\dt+\int_{\R}|\tilde{w}^{\prime}(t)|^2\dt\leq h_G(q_-,q_+)+\e.
\end{equation*}
Without loss of generality we may assume that $\{t\in\R:\;\tilde{w}(t)=0\}$ is at most a singleton, otherwise a construction as in (\ref{cutout}) reduces the energy. Moreover, by the existence of the limits at $\pm \infty$, we find $t_{\e}>0$ such that
\begin{align}
|\tilde{w}(t)-q_-|&\leq\e \quad\forall t\leq -t_{\e},\label{atplusinf}\\
|\tilde{w}(t)-q_+|&\leq \e \quad\forall t\geq t_{\e},\label{atminusinf}.
\end{align}
Approximating $\tilde{u}$ in $H^2((-t_{\e}-3,t_{\e}+3))$ (note that $C_c^{\infty}(\R,S^2)$ is dense in $H^2((a,b),S^2)$ for every bounded interval $(a,b)$) we can assume that $\tilde{u}$ is smooth in $[-t_{\e}.t_{\e}]$, that (\ref{atplusinf}), (\ref{atminusinf}) still hold at least at $t=\pm t_{\e}$ respectively and
\begin{equation*}
\int_{-t_{\e}}^{t_{\e}}(|\tilde{w}(t)|^2-1)^2+\frac{G(\tilde{w}(t))}{2}\dt+\int_{-t_{\e}}^{t_{\e}}|\tilde{w}^{\prime}(t)|^2\dt\leq h_G(q_-,q_+)+2\e
\end{equation*}
since approximation of $\tilde{u}$ in $H^2((-t_{\e}-3,t_{\e}+3))$ implies approximation of $\tilde{w}$ in $H^1((-t_{\e},t_{\e}))$ and the discontinuity set of $G$ can be neglected. Preserving at least a bounded weak second derivative of $\tilde{u}$ in $(-t_{\e},t_{\e})$ we can again assume that $\{\tilde{w}=0\}$ is at most a singleton. 

We now modify the function $\tilde{u}$ where $|t|>t_{\e}$ in the following way: From (\ref{at+inf}) it follows that $|\tilde{u}^{\prime}(t_{\e})|\neq 0$. Consider then the matrix 
\begin{equation*}
R_{\e}=
\begin{pmatrix}
\\
\tilde{u}(t_{\e}) & \frac{\tilde{u}^{\prime}(t_{\e})}{|\tilde{u}^{\prime}(t_{\e})|} & v\\
&
\end{pmatrix},
\end{equation*}
where $v\in S^2$ is the vector that makes the matrix orthogonal. Now we take a $C^2$-function $\gamma:[t_{\e},t_{\e}+1]\rightarrow\R$ with the following properties:
\begin{enumerate}
	\item[(i)] $\gamma(t_{\e})=\gamma(t_{\e}+1)=0,$
	\item[(ii)] $\gamma^{\prime}(t_{\e})=|\tilde{u}^{\prime}(t_{\e})|-1,$
	\item[(iii)] $\gamma^{\prime}(t_{\e}+1)=\gamma^{\prime\prime}(t_{\e}+1)=0$
\end{enumerate}
and extend it to $0$ for $t>t_{\e}+1$. This extension (not relabeled) is obviously $C^2$-regular. $\gamma$ can be chosen such that $\max\{\|\gamma\|_{\infty},\|\gamma^{\prime}\|_{\infty},\|\gamma^{\prime\prime}\|_{\infty}\}\leq C\e$ with a positive constant independent of $\tilde{u}$ and $\e$. The modification $u$ of $\tilde{u}$ on $(-t_{\e},t_{\e}+2)$ now is defined as
\begin{equation*}
u(t):=
\begin{cases}
\tilde{u}(t) &\mbox{if $t\leq t_{\e}$,}\\
R_{\e}(\cos(t-t_{\e}+\gamma(t)),\sin(t-t_{\e}+\gamma(t)),0) &\mbox{if $t\in (t_{\e},t_{\e}+2)$.}
\end{cases}
\end{equation*}
Note that $u\in H^2((-t_{\e},t_{\e}+2),S^2)$ and its weak second derivative is bounded (but not necessarily continuous at $t_{\e}$). A straightforward calculation shows that for $t\in (t_{\e},t_{\e}+2)$ we have
\begin{align*}
&|u^{\prime}(t)|=|1+\gamma^{\prime}(t)|,\\
&|u(t)\times u^{\prime\prime}(t)|\leq |\gamma^{\prime\prime}(t)|.
\end{align*}
Moreover we have 
\begin{equation*}
u(t_{\e}+2)\times u^{\prime}(t_{\e}+2)=v=\frac{\tilde{u}(t_{\e})\times \tilde{u}^{\prime}(t_{\e})}{|\tilde{u}^{\prime}(t_{\e})|},
\end{equation*}
so that by the choice of $t_{\e}$
\begin{equation}\label{extensionend}
|(u(t_{\e}+2)\times u^{\prime}(t_{\e}+2))-q_+|\leq \e + |(\tilde{u}(t_{\e})\times \tilde{u}^{\prime}(t_{\e}))-\frac{\tilde{u}(t_{\e})\times \tilde{u}^{\prime}(t_{\e})}{|\tilde{u}^{\prime}(t_{\e})|}|\leq 2\e.
\end{equation}
For $t\in (t_{\e},t_{\e}+2)$ we also have by the zero-homogeneity of $G$ that
\begin{equation*}
G(u(t)\times u^{\prime}(t))=G(v)=G(w(t_{\e}))\leq c_G(\e).
\end{equation*}

We now use the same method as in the second part of the proof of Lemma \ref{connectpoints} to construct a further $C^2$-extension on $[t_{\e}+2,t_{\e}+3]$ that ends in the constant rotation with velocity $1$ in the plane perpendicular to $q_+$. The same procedure can be applied at $-t_{\e}$. Keeping in mind (\ref{extensionend}) and Lemma \ref{connectpoints} we have constructed a new function $u\in H^2_{loc}(\R,S^2)$ with bounded weak second derivative and $t^{\prime}_{\e}>0$ such that
\begin{align*}
u(t)&=R_1(\cos(t),\sin(t),0)\quad\forall t\leq -t^{\prime}_{\e},\\
u(t)&=R_2(\cos(t),\sin(t),0)\quad\forall t\geq t^{\prime}_{\e},
\end{align*}
with $R_1,R_2\in SO(3)$, the function $w=u\times u^{\prime}$ is admissible in the definition of $h_G(q_-,q_+)$ and there exists $C_{\e}>0$ with $\lim_{\e\to 0}C_{\e}=0$ such that
\begin{equation*}
\int_{\R}(|w(t)|^2-1)^2+\frac{G(w(t))}{2}\dt+\int_{\R}|w^{\prime}(t)|^2\dt\leq h_G(q_-,q_+)+C_{\e}.
\end{equation*}
Moreover, $u^{\prime\prime}$ is continuous except in at most three points. The claim follows by the arbitrariness of $\e$, since the other inequality is trivial.
\end{proof}

Depending on the behaviour of the sequence $\mu_n$ we have different variational limits.

\begin{theorem}\label{mainangle}
Let $H_n^{p}:L^1(I,\R^3)\rightarrow [0,+\infty]$ be defined as in (\ref{penalzvariable}). Assume that $\beta_n:=\frac{\Ln}{\sqrt{\dn}}\to 0$ and let $p_n:=\frac{\mu_n}{\sqrt{2}\Ln\dn^{\frac{3}{2}}}$. Then the $\Gamma$-limit of the functionals $\frac{H_n^{p}}{\sqrt{2}\Ln\dn^{\frac{3}{2}}}$ with respect to the strong $L^1$-topology is given by the functional $H^{p}:L^1(I,\R^3)\rightarrow [0,+\infty]$, depending on the following four cases: 
\begin{itemize}
\item[(i)] $\lim_n p_n=p<+\infty$:
\begin{equation*}
H^{p}(z)=
\begin{cases}
p\int_IG(z(t))\dt &\mbox{if $z\in L^1(I,S^2)$,}
\\
+\infty &\mbox{otherwise.}
\end{cases}
\end{equation*}

\item[(ii)] $\lim_n p_n=+\infty,\,\lim_n p_n\beta_n=0$:
\begin{equation*}
H^p(z)=
\begin{cases}
0 &\mbox{if $z\in L^1(I,Q_k)$,}
\\
+\infty &\mbox{otherwise.}
\end{cases}
\end{equation*}

\item[(iii)] $\lim_n p_n=+\infty,\,\lim_n p_n\beta_n=1$:
\begin{equation*}
H^{p}(z)=
\begin{cases}
\sum_{t\in S(z)}h_G(q_-,q_+) &\mbox{if $z\in BV(I,Q_k)$,}
\\
+\infty &\mbox{otherwise,}
\end{cases}
\end{equation*}
where $q_-$ and $q_+$ are the left and right limit of $z$ at a discontinuity point $t$.
\item [(iv)] $\lim_n p_n=+\infty,\,\lim_n p_n\beta_n=+\infty$:
\begin{equation*}
H^{p}(z)=
\begin{cases}
\frac{8}{3}\#S(z) &\mbox{if $z\in BV(I,Q_k)$,}
\\
+\infty &\mbox{otherwise.}
\end{cases}
\end{equation*}
\end{itemize}
\end{theorem}

\begin{proof}

(i): Let $z_n$ converge to $z$ in $L^1(I,\R^3)$.  By Proposition \ref{weakcompactness} we know that $z\in L^1(I, S^2)$. We now show that
\begin{equation*}
\lim_n p_n\sum_{i\in R_n(I)}\Ln G(z_n^i)=p\int_IG(z(t))\dt.
\end{equation*} 
Up to a subsequence, we can assume that $z_n(t)\to z(t)\in S^2$ for almost every $t\in I$. In particular, by the continuity of $G$ in $\R^3\backslash\{0\}$ we may assume that $G(z_n(t))\to G(z(t))$ for almost every $t\in I$. Moreover $\|G\|_{\infty}<+\infty$, so that by dominated convergence the above limit relation holds true for the whole sequence. 
With the above limit, the upper bound follows considering the same recovery sequence as in Theorem \ref{main1}, while the lower bound is obvious since the remaining part of the energy is nonnegative.

(ii): To prove the lower bound, note that since $p_n\to +\infty$, the penalization forces any $L^1$-converging sequence with bounded energy to have a limit $z\in L^1(I,Q_k)$. For the upper bound we can use exactely the same construction as in the proof of Theorem \ref{main1} upon noticing that the assumption $p_n\beta_n\to 0$ is enough to kill the penalization term after rescaling.

(iii): {\bf Lower bound.} Without loss of generality let $C_n(I,\R^3)\ni z_n\to z$ in $L^1(I,\R^3)$ such that
\begin{equation*}
\liminf_n\frac{H_n^{p}(z_n)}{\sqrt{2}\Ln\dn^{\frac{3}{2}}}=\lim_n\frac{H_n^{p}(z_n)}{\sqrt{2}\Ln\dn^{\frac{3}{2}}}\leq C<+\infty.
\end{equation*}
From Proposition \ref{penalcompact} we know that $z\in BV(I,Q_k)$. Passing to a further subsequence (not relabeled) we can assume that $z_n$ converges to $z$ almost everywhere. Let $t_1<\dots<t_l$ be the jumpset of $z$. Let $\alpha>0$ be such that $[-2\alpha+t_m,t_m+2\alpha]\cap[-2\alpha+t_j,t_j+2\alpha]=\emptyset$ for  all $j\neq m$. Fix $t_m$ and set $q_{\pm}:=z(t_m\pm\alpha)$. Let now $u_n \in T_n^{-1}(z_n)$ be such that $H_n^{p}(u_n)=H_n^{p}(z_n)$. By (\ref{liminfbound}) it is enough to show that
\begin{equation*}
\liminf_n F^{G,m}_n(z_n)\geq h_G(q_-,q_+),
\end{equation*}
where
\begin{align*}
F^{G,m}_n(z_n)=&\frac{\sqrt{2\dn}}{\Ln}\sum_{|\Ln i-t_m|<2\alpha}\Ln \left\{\left(\left|\frac{u_n^{i+1}-u_n^i}{\sqrt{2\dn}}\right|^2-1\right)^2+\frac{G(z_n^i)}{{2}}\right\}
\\
&+\frac{\Ln}{\sqrt{2\dn}}\sum_{|\Ln i-t_m|<2\alpha}\Ln\left|\frac{z_n^{i+1}-z_n^i}{\Ln}\right|^2.
\end{align*}
Using (\ref{uniformclose}), again we can write 
\begin{equation*}
\frac{\sqrt{2\dn}}{\Ln}\sum_{|\Ln i-t_m|<2\alpha}\Ln \left(\left|\frac{u_n^{i+1}-u_n^i}{\sqrt{2\dn}}\right|^2-1\right)^2=\frac{\sqrt{2\dn}}{\Ln}\sum_{|\Ln i-t_m|<2\alpha}\Ln \left(|\beta_n^i z_n^i|^2-1\right)^2
\end{equation*}
with a function $\beta_n\in C_n(I,\R)$ converging uniformly to $1$. 
Let $u^a_n$ be the piecewise affine interpolation of $u_n$ and observe that by Proposition \ref{weakcompactness} and \eqref{uniformclose} we have
\begin{equation}\label{boundscaled}
\|(u^a_n)^{\prime}\|_\infty\le C \frac{\sqrt{\dn}}{\Ln}\,.
\end{equation}

Now we define the rescaled piecewise affine and piecewise constant functions $\tilde{u}^a_n:\frac{\sqrt{2\dn}}{\Ln}(-2\alpha,2\alpha)\rightarrow S^2$ and $\tilde{z}_n:\frac{\sqrt{2\dn}}{\Ln}(-2\alpha,2\alpha)\rightarrow\R^3$ setting
\begin{align*}
\tilde{u}^a_n(t)&=u_n^a\left(\frac{\Ln}{\sqrt{2\dn}}t+t_m\right),\quad\tilde{z}_n(t)=z_n\left(\frac{\Ln}{\sqrt{2\dn}}t+t_m\right).
\end{align*}
Note that $\tilde{z}_n$ is constant on intervals of length $\sqrt{2\dn}$, and that \eqref{continuity} implies
\begin{equation}\label{quasifinito}
|\tilde{z}_n(t_1)-\tilde{z}_n(t_2)|^2 \le C\sqrt{\dn}
\end{equation}
whenever $|t_1-t_2|\le \sqrt{2\dn}$. Moreover, by the definition of piecewise interpolations, for almost every $t$ it holds that
\begin{equation}\label{prepareconstraint}
\tilde{z}_n(t)=\tilde{u}^a_n(t)\times(\tilde{u}^a_n)^{\prime}(t).
\end{equation}
while \eqref{boundscaled} implies that
\begin{equation}\label{boundu}
\|(\tilde u^a_n)^{\prime}\|_\infty\le C\,.
\end{equation}
We also notice that by Remark \ref{affineclose} we have
\begin{equation}\label{affine-constant}
\|\tilde{z}_n-\tilde{z}^a_n\|_\infty \to 0\,,
\end{equation}
where $\tilde{z}_n^a$ denotes the piecewise affine interpolation of $\tilde{z}_n$. Rewriting the energy in terms of these interpolations leads to
\begin{align}
 \nonumber\liminf_{n} F^{G,m}_n(z_n)&\geq
\\
\liminf_n\frac{\sqrt{2\dn}}{\Ln}\int_{t_m-\alpha}^{t_m+\alpha}\left(|\beta_n(s)z_n(s)|^2-1\right)^2+&\frac{G(z_n(s))}{2}\,\mathrm{d}s+\frac{\Ln}{\sqrt{2\dn}}\int_{t_m-\alpha}^{t_m+\alpha}|(z_n^a)^{\prime}(s)|^2\,\mathrm{d}s =\label{tointegral}
\\
\nonumber\liminf_n\int_{-\alpha\frac{\sqrt{2\dn}}{\Ln}}^{\alpha\frac{\sqrt{2\dn}}{\Ln}}\left(|\tilde{\beta}_n(t)\tilde{z}_n(t)|^2-1\right)^2+&\frac{G(\tilde{z}_n(t))}{2}\,\mathrm{d}t
+\int_{-\alpha\frac{\sqrt{2\dn}}{\Ln}}^{\alpha\frac{\sqrt{2\dn}}{\Ln}}|(\tilde{z}_n^a)^{\prime}(t)|^2\,\mathrm{d}t,
\end{align}
where we used the change of variables $s=\frac{\sqrt{2\dn}}{\Ln}t+t_m$ and the function  $\tilde{\beta}_n$ is defined as $\tilde{\beta}_n(t)=\beta_n(s)$, so that it still converges uniformly to $1$. Let $0<\eta<\frac14 \min\{{\rm dist }(q_l, q_m): q_l\,,q_m \in Q_k\}$. By (\ref{quasifinito}) there exists $t^n_*\in\frac{\sqrt{2\dn}}{\Ln}(-\alpha,\alpha)$ such that $\min_{\pm}|\tilde{z}_n(t^n_*)-q_{\pm}|>\eta$. Then let $t_-^n\leq t^n_*\leq t_+^n$ be respectively the largest and the smallest point such that
\begin{equation}\label{close}
|\tilde{z}_n(t_{\pm}^n)-q_{\pm}|\leq\eta.
\end{equation}
We now prove the following claim.

{\it Claim}: to the given $\eta$ and for all $n$ large enough, it exists an interval $I_n=(\tau^-_n, \tau^+_n)$ having equibounded measure with respect to $n$, three sequences $w_n$, $\bar w_n$ and $\bar u_n$ equibounded in $L^\infty (I_n; \R^3)$, $W^{1,\infty}(I_n; \R^3)$, and $W^{1,\infty}(I_n; S^2)$, respectively, and a sequence $\hat \beta_n$ satisfying the following properties: 
\begin{eqnarray}
&\displaystyle \nonumber
\vphantom{\liminf_n\int_{I_n}}|w_n(\tau^-_n)-q_-|\le \eta\,,\quad |w_n(\tau^+_n)-q_+|\le \eta\,,\quad \|\hat \beta_n-1\|_{L^\infty(I_n; \R)}\to 0\\
&\displaystyle \nonumber
\vphantom{\liminf_n\int_{I_n}}\|w_n-\bar w_n\|_{L^\infty(I_n; \R^3)}\to 0\,,\quad  w_n =\bar u_n \times (\bar u_n)^{\prime}\\
&\displaystyle
\liminf_{n} F^{G,m}_n(z_n) \ge \\
&\displaystyle
\nonumber\liminf_n\int_{I_n}\left(|\hat{\beta}_n(t)w_n(t)|^2-1\right)^2+\frac{G(w_n(t))}{2}\,\mathrm{d}t 
+\int_{I_n}|(\bar w_n)^{\prime}(t)|^2\,\mathrm{d}t- (2k-2)C_{2\eta}\,,\label{claim-main}
\end{eqnarray}
for some constant $C_{2\eta} \to 0$ when $\eta \to 0$. All the constructed sequences as well as the interval $I_n$ will be depending on $\eta$, but we omit this dependence in order to ease notation.

In order to prove this claim, we follow an algorithmic construction. We first notice that the set
\begin{equation*}
J^{\eta}_n:=\{t\in\frac{\sqrt{2\dn}}{\Ln}(-\alpha,\alpha):\;\text{dist}(\tilde{z}_n(t),Q_k)\geq\eta \}
\end{equation*}
has finite measure uniformly in $n$, due to the energy bound.

{\it Step $0$}: if $|t^+_n-t^-_n|$ is equibounded, the claim is proved with $I_n=(t^-_n, t^+_n)$, $w_n=\tilde{z}_n$, $\bar w_n= \tilde{z}_n^a$, $\bar u_n=\tilde{u}_n^a$, and $\hat \beta_n=\tilde \beta_n$, due to \eqref{prepareconstraint}, \eqref{boundu}, \eqref{tointegral}, \eqref{affine-constant}, and \eqref{close}. If not, we proceed to Step $1$, upon noticing that, if this is the case, the inclusion $(t^-_n, t^+_n) \subseteq J^{\eta}_n$ is not satisfied, otherwise $|t^+_n-t^-_n|$ would be equibounded.

{\it Step $1$}: we set $\tau^{-}_n= t^-_n$, $\tau^{0,+}_n= t^+_n$, and $I^0_n:=(\tau^{-}_n, \tau^{0,+}_n)$. We define the functions $\beta^0_n=\tilde \beta_n$, $w^0_n=\tilde{z}_n$, $\bar w^0_n= \tilde{z}_n^a$ and $\bar u^0_n=\tilde{u}_n^a$. Let us number the points $q_1,\dots,q_{2k-2}\in Q_k\setminus\{q_{\pm}\}$. 
By construction of $\tau^{-}_n$ and $\tau^{0,+}_n$, for all $t \in (\tau^{-}_n, \tau^{0,+}_n) $ it holds
\begin{equation}\label{entfernung}
\min\{|w^0_n(t)-q_-|\,,|w^0_n(t)-q_+|\}>\eta
\end{equation}
Since $(\tau^{-}_n, \tau^{0,+}_n)$ is not contained in $J^{\eta}_n$, also using \eqref{quasifinito}, it must exist a minimal time $t^{1,-}_n \in (\tau^{-}_n, \tau^{0,+}_n)$ such that
$$
{\rm dist}(w^0_n(t^{1,-}_n), Q_k)\le \eta\,;
$$
due to \eqref{entfernung} and our choice of $\eta$, to the time $t^{1,-}_n$ it corresponds a uniquely determined point $q_{1} \in Q_k\setminus\{q_{\pm}\}$ such that
$$
|w^0_n(t^{1,-}_n)- q_{1}|\le \eta\,.
$$
Notice that by construction $t^{1,-}_n$ must be the left endpoint of an interval where $w^0_n$ is constant, therefore 
\begin{equation}\label{aff=konst1}
w^0_n(t^{1,-}_n)=\bar w^0_n(t^{1,-}_n)\,.
\end{equation}
We define 
$$
t^{1,+}_n:=\sup\{t\in (t^{1,-}_n,\tau^{0,+}_n):\;|w^0_n(t)-q_{1}|<\frac54\eta\}\,,
$$
again noticing that, since $|w^0_n(\tau^{0,+}_n)-q_+|\le \eta$, the above supremum is well defined and strictly small than $\tau^{0,+}_n$ thanks to \eqref{quasifinito}. By construction $t^{1,+}_n$ must be the left endpoint of an interval where $w^0_n$ is constant, therefore
\begin{equation}\label{aff=konst2}
w^0_n(t^{1,+}_n)=\bar w^0_n(t^{1,+}_n)\,.
\end{equation}
Furthermore, due to \eqref{quasifinito}, one has for $n$ large enough that
$$
|w^0_n(t^{1,+}_n)- q_{1}|\le 2\eta\,.
$$

If now $|t^{1,+}_n-t^{1,-}_n|>3+4\pi$, we consider $t^*\leq 3+4\pi$ and the function $u \in C^2([0, t^*]; S^2)$ and $w=u\times u^{\prime} \in H^1([0,t^*], \R^3)$ given by Lemma \ref{connectpoints} with data $w_0=w^0_n(t^{1,-}_n)$ and $w_1=w^0_n(t^{1,+}_n)$ as well as $u_0=\bar u^0_n(t^{1,-}_n)$ and $u_1=\bar u^0_n(t^{1,+}_n)$. We set $\tau^{1,-}_n=t^{1,-}_n + t^*$, $\tau^{1,+}_n=\tau^{0,+}_n -[(t^{1,+}_n-t^{1,-}_n)-t^*]$, and we define
\begin{align*}
w^1_n(t)=
\begin{cases}
w^0_n(t) &\mbox{if $t\leq t^{1,-}_n$,}
\\
w(t-t^{1,-}_n) &\mbox{if $t^{1,-}_n\leq t \leq\tau^{1,-}_n$,}
\\
w^0_n(t)(t-t^{1,+}_n+t^{1,-}_n+t^*) &\mbox{if $\tau^{1,-}_n\leq t \leq\tau^{1,+}_n$,}
\end{cases}
\\
\bar w^1_n(t)=
\begin{cases}
\bar w^0_n(t) &\mbox{if $t\leq t^{1,-}_n$,}
\\
w(t-t^{1,-}_n) &\mbox{if $t^{1,-}_n\leq t \leq\tau^{1,-}_n$,}
\\
\bar w^0_n(t)(t-t^{1,+}_n+t^{1,-}_n+t^*) &\mbox{if $\tau^{1,-}_n\leq t \leq\tau^{1,+}_n$,}
\end{cases}
\\
\bar u^1_n(t)=
\begin{cases}
\bar u^0_n(t) &\mbox{if $t\leq t^{1,-}_n$,}
\\
u(t-t^{1,-}_n) &\mbox{if $t^{1,-}_n\leq t \leq\tau^{1,-}_n$,}
\\
\bar u^0_n(t)(t-t^{1,+}_n+t^{1,-}_n+t^*) &\mbox{if $\tau^{1,-}_n\leq t \leq\tau^{1,+}_n$,}
\end{cases}
\\
\bar \beta^1_n(t)=
\begin{cases}
\bar \beta^0_n(t) &\mbox{if $t\leq \tau^{1,-}_n$,}
\\
\bar \beta^0_n(t)(t-t^{1,+}_n+t^{1,-}_n+t^*) &\mbox{if $\tau^{1,-}_n\leq t \leq\tau^{1,+}_n$.}
\end{cases}
\end{align*}
We set $I^1_n:=(\tau^{0,-}_n, \tau^{1,+}_n)$. The above construction preserves continuity, and then Sobolev regularity of $\bar w^1_n$ because of \eqref{aff=konst1} and \eqref{aff=konst2}. 

The bound on the norms of $w^1_n$, $\bar w^1_n$ and $\bar u^1_n$ on $I^1_n$ are satisfied by construction, as well as the relation $w^1_n=\bar u^1_n \times (\bar u^1_n)^{\prime}$ and it clearly holds
\begin{equation}\label{lunendlich}
\|w^1_n(t)-\bar w^1_n(t)\|_{L^\infty(I^1_n;\R^3)}=\|w^0_n-\bar w^0_n\|_{L^\infty(I^0_n;\R^3)}
\end{equation}
as well as
\begin{equation}\label{lunendlich2}
\|\beta^1_n-1\|_{L^\infty(I^1_n)}\le\|\beta^0_n-1\|_{L^\infty(I^0_n)}\,.
\end{equation}
By construction we have
\begin{equation*}
w^1_n(\tau^{-}_n)=\tilde{z}_n(t^-_n)\,,\quad w^1_n(\tau^{1,+}_n)=\tilde{z}_n(t^+_n)\,.
\end{equation*}
It holds furthermore
\begin{equation}\label{i1n}
\left|I^1_n\right| \leq |J^{\eta}_n|+1\cdot (3+4\pi)+|\tau^{1,+}_n-\tau^{1,-}_n|\,.
\end{equation}
Since in $(\tau^{-}_n, t^{1,-}_n)$ and in $(\tau^{1,-}_n, \tau^{1,+}_n)$  we contructed the functions $w^1_n(t)$, $\bar w^1_n(t)$, $\bar u^1_n(t)$, and $\beta^1_n(t)$ are constructed by applying the same translation to the functions $w^0_n(t)$, $\bar w^0_n(t)$, $\bar u^0_n(t)$, and $\beta_n^0(t)$ it clearly holds
\begin{align*}
\liminf_n\int_{I^1_n\setminus (t^{1,-}_n,\tau^{1,-}_n)}\left(|\beta^1_n(t)w^1_n(t)|^2-1\right)^2+\frac{G(w^1_n(t))}{2}\,\mathrm{d}t
+\int_{I^1_n\setminus (t^{1,-}_n,\tau^{1,-}_n)}|(\bar w^1_n)^{\prime}(t)|^2\,\mathrm{d}t \le
\\
\liminf_n\int_{I^0_n}\left(|\beta^0_n(t)w^0_n(t)|^2-1\right)^2+\frac{G(w^0_n(t))}{2}\,\mathrm{d}t
+\int_{I^0_n}|(\bar w^0_n)^{\prime}(t)|^2\,\mathrm{d}t\,.
\end{align*}
Since $(t^{1,-}_n,\tau^{1,-}_n)$ has equibounded measure, using the uniform convergence of $\beta^1_n$ to $1$ and Lemma \ref{connectpoints} we get
\begin{align*}
\limsup_n\int_{t^{1,-}_n}^{\tau^{1,-}_n}\left(|\beta^1_n(t)w^1_n(t)|^2-1\right)^2+\frac{G(w^1_n(t))}{2}\,\mathrm{d}t
+\int_{t^{1,-}_n}^{\tau^{1,-}_n}|(\bar w^1_n)^{\prime}(t)|^2\,\mathrm{d}t \le
\\
\int_0^{t^*}\left(|w(t)|^2-1\right)^2+\frac{G(w(t))}{2}\dt+\int_0^{t^*}|w^{\prime}(t)|^2\dt\leq C_{2\eta}\,,
\end{align*}
so that summing the two above inequalities we arrive at
\begin{align}
\nonumber
\liminf_n\int_{I^1_n}\left(|\beta^1_n(t)w^1_n(t)|^2-1\right)^2+\frac{G(w^1_n(t))}{2}\,\mathrm{d}t
+\int_{I^1_n}|(\bar w^1_n)^{\prime}(t)|^2\,\mathrm{d}t \le
\\
\liminf_n\int_{I^0_n}\left(|\beta^0_n(t)w^0_n(t)|^2-1\right)^2+\frac{G(w^0_n(t))}{2}\,\mathrm{d}t
+\int_{I^0_n}|(\bar w^0_n)^{\prime}(t)|^2\,\mathrm{d}t+C_{2\eta}\,.\label{iterative}
\end{align}

If instead $|t^{1,+}_n-t^{1,-}_n|\le 3+4\pi$ we simply set $\tau^{1,-}_n=t^{1,+}_n$, $\tau^{1,+}_n=\tau^{0,+}_n$. With this choice, \eqref{i1n} clearly holds and all the properties \eqref{lunendlich}, \eqref{lunendlich2}, and \eqref{iterative} are satisfied by simply setting $w^1_n=w^0_n$, $\bar w^1_n=\bar w^0_n$, $\bar u^1_n= \bar u^0_n$ and $\beta^1_n=\beta_0^n$.

We now define
$$
J^{\eta,1}_n:=\{t\in I^1_n:\;\text{dist}(w^1_n(t),Q_k)\geq\eta \}
$$
and observe that by construction
\begin{equation}\label{boundJeta}
|J^{\eta,1}_n|\le |J^{\eta}_n|+1\cdot (3+4\pi)\,.
\end{equation}
If now $(t^{1,-}_n, \tau^{1,+}_n) \subseteq J^{\eta,1}_n$ the construction stops, if not we go to the next step.

{\it Step $\ell$}: first observe that by construction, it holds
\begin{equation}\label{stop-criterium}
\text{dist}(w^{\ell-1}_n(t), \{q_-, q_+, q_1,\dots, q_{\ell-1}\})>\eta
\end{equation}
for all $t \in (\tau^{\ell-1,-}_n, \tau^{\ell-1,+}_n)$.
One sets $w^{\ell}_n(t)= w^{\ell-1}_n(t)$ for all $t \in (\tau^{-}_n, \tau^{\ell-1,-}_n]$, etc.,  and performs the construction described in Step 1 in the interval $(\tau^{\ell-1,-}_n, \tau^{\ell-1,+}_n)$. For $I^{\ell}_n=(\tau^{-}_n, \tau^{\ell,+}_n)$ one can see that all the properties \eqref{lunendlich}, \eqref{lunendlich2}, \eqref{i1n},\eqref{iterative}, and \eqref{boundJeta} hold with $\ell$ in place of $1$ and $\ell-1$ in place of $0$.

{\it Proof of the claim}: By \eqref{stop-criterium} the construction ends in a finite number $\bar \ell$ of steps with $\bar \ell \leq (2k-2)$. We then set $I_n:=I^{\bar \ell}_n$, and consequently $\tau^+_n=\tau^{\bar \ell, +}_n$. We set $w_n:=w^{\bar \ell}_n$, $\bar w_n:= \bar{w}^{\bar \ell}_n$, $\bar{u}_n:=\bar{u}^{\ell}_n$ and $\hat \beta_n:=\beta^{\bar \ell}_n$. By construction we have $w_n=\bar{u}_n\times (\bar{u}_n)^{\prime}$ and 
\begin{equation*}
w_n(\tau^{-}_n)=\tilde{z}_n(t^-_n)\,,\quad w_n(\tau^{+}_n)=\tilde{z}_n(t^+_n)\,.
\end{equation*}
Since at the final step it holds  $(\tau^{\bar \ell, -}_n,\tau^{\bar \ell, +}_n)\subseteq J^{\eta,\bar \ell}_n$, combining the inequalities
$$
\left|I_n\right| \leq |J^{\eta}_n|+\bar \ell\cdot (3+4\pi)+|\tau^{\bar \ell,+}_n-\tau^{\bar \ell,-}_n|
$$
and 
$$
|\tau^{\bar \ell,+}_n-\tau^{\bar \ell,-}_n|\le |J^{\eta,\bar \ell}_n|\le |J^{\eta}_n|+\bar \ell(3+4\pi)
$$
we get that $I_n$ has equibounded measure. All the other properties in \eqref{claim-main} follow by iteratively applying \eqref{lunendlich}, \eqref{lunendlich2}, \eqref{i1n},\eqref{iterative}, and \eqref{boundJeta} with $\ell$ in place of $1$ and $\ell-1$ in place of $0$ for all $1\leq \ell \leq \bar \ell$.

{\it Conclusion of the lower bound}: Possibly after a translation not changing the energy we can assume that $I_n=[0,\tau_n]$ with $\tau_n$ equibounded. Since  $\hat{\beta}_n$ converges uniformly to $1$ we deduce from \eqref{claim-main} that
\begin{align}
\nonumber\liminf_{n} F^{G,m}_n(z_n)&
\geq\\
\liminf_n\int_0^{\tau_n}\left(|w_n(t)|^2-1\right)^2+&\frac{G(w_n(t))}{2}\,\mathrm{d}t
+\int_0^{\tau_n}|(\bar w_n)^{\prime}(t)|^2\,\mathrm{d}t-(2k-2)C_{2\eta}.\label{compact}
\end{align}
Let $\tau=\liminf_n \tau_n$. Since $\|w_n-\bar w_n\|_\infty \to 0$ and the bound on the $H^1$-norm of $\bar w_n$ is independent of $n$, up to subsequences $w_n$ and $\bar w_n$ are locally uniformly converging to a function $w \in H^1((0,\tau); \R^3)$. By this convergence we have
$$
|w(0)-q_-|\le \eta
$$
while, by means of a simple equicontinuity argument, we get $\liminf_n|w(\tau)-w_n(\tau_n)|=0$ and therefore
$$
|w(\tau)-q_+|\le \eta\,.
$$
Since $w_n=\bar u_n \times (\bar u_n)^{\prime}$ and $\bar u_n$ is an equibounded sequence in $W^{1,\infty}(I;S^2)$, it exists $\bar u \in W^{1,\infty}(I;S^2)$ such that $w=\bar u\times (\bar u)^{\prime}$. Furthermore by lower semincontinuity we have
\begin{align}
\nonumber\liminf_n\int_0^{\tau_n}\left(|w_n(t)|^2-1\right)^2+\frac{G(w_n(t))}{2}\,\mathrm{d}t
+\int_0^{\tau_n}|(\bar w_n)^{\prime}(t)|^2\,\mathrm{d}t\geq\\
\int_0^{\tau}\left(|w(t)|^2-1\right)^2+\frac{G(w(t))}{2}\,\mathrm{d}t
+\int_0^{\tau}|w^{\prime}(t)|^2\,\mathrm{d}t\label{limit-int}
\end{align}
Using Lemma \ref{connectpoints} we can now extend $w$ to a function in $H^{\times}_{q_-,q_+}$ such that
\begin{align*}
\displaystyle
\int_{\R}\left(|w(t)|^2-1\right)^2+\frac{G(w(t))}{2}\,\mathrm{d}t
+\int_{\R}|w^{\prime}(t)|^2\,\mathrm{d}t\le \\[5pt]
\displaystyle
\int_0^{\tau}\left(|w(t)|^2-1\right)^2+\frac{G(w(t))}{2}\,\mathrm{d}t
+\int_0^{\tau}|w^{\prime}(t)|^2\,\mathrm{d}t+2 C_{\eta}\,.
\end{align*}
Since the first member of the inequality is by definition larger than $h_G(q_-, q_+)$, combining this with \eqref{compact} and \eqref{limit-int} we finally get
$$
\liminf_{n} F^{G,m}_n(z_n)\geq h_G(q_-, q_+)-2 C_{\eta}-(2k-2)C_{2\eta}
$$
and the lower bound follows by letting $\eta \to 0$.

{\bf Upper bound:} In order to prove the upper bound, as usual we provide a local construction and restrict to the case of $\#S(z)=1$, so without loss generality we may assume $z=q_1\mathds{1}_{[0,\frac{1}{2}]}+q_2\mathds{1}_{(\frac{1}{2},1]}$ for some $q_1\neq q_2$. Given $\e>0$ we find a function $w=u\times u^{\prime}$ as in Lemma \ref{regularclass} such that
\begin{equation*}
\int_{\R}(|{w}(t)|^2-1)^2+\frac{G({w}(t))}{2}\dt+\int_{\R}|{w}^{\prime}(t)|^2\dt\leq h_G(q_1,q_2)+\e.
\end{equation*}
The interpolation of the function $u$ in order to construct a recovery sequence now is the same as in the proof of Theorem \ref{main1}. Note that we can pass to the limit in (\ref{limsup1}) again, and the limit in (\ref{limsup2}) follows since $u^{\prime\prime}$ is bounded and has only finitely many discontinuities.

Therefore, it only remains to prove that
\begin{equation*}
\limsup_n\frac{\sqrt{2\dn}}{\Ln}\sum_{i\in R_n(I)}\Ln \frac{G(z_n^i)}{2}\leq\int_{\R}\frac{G(w(t))}{2}\dt.
\end{equation*}
This is done as usual with a change of variables and using the fact that $\{w(t)=0\}$ is finite, so that the discontinuity of $G$ in the origin can be neglected.

(iv): To prove the lower bound, let $C_n(I,\R^3)\ni z_n\to z$ in $L^1(I,\R^3)$ such that
\begin{equation*}
\liminf_n\frac{H_n^{p}(z_n)}{\sqrt{2}\Ln\dn^{\frac{3}{2}}}=\lim_n\frac{H_n^{p}(z_n)}{\sqrt{2}\Ln\dn^{\frac{3}{2}}}\leq C<+\infty.
\end{equation*}
From Proposition \ref{penalcompact} we know that $z\in BV(I,Q_k)$. Passing to a further subsequence (not relabeled) we can assume that $z_n$ converges to $z$ almost everywhere. Let $t_1<\dots<t_l$ be the jumpset of $z$. Let $\alpha>0$ be such that $[-2\alpha+t_m,t_m+2\alpha]\cap[-2\alpha+t_j,t_j+2\alpha]=\emptyset$ for  all $j\neq m$. Fix $t_m$ and set $q_{\pm}:=z(t_m\pm\alpha)$. We now prove that
\begin{equation}\label{backtothm}
\text{dist}(z_n^a,L_k)\to 0 \quad\text{uniformly on $(t_m-\alpha,t_m+\alpha)$.}
\end{equation}

To this end, we fix $\eta>0$ and assume by contradiction, that for every $n$ there exists $\tau_n\in (t_m-\alpha,t_m+\alpha)$ with $\text{dist}(z_n^a(\tau_n),L_k)\geq\eta$. Using $L^1$-convergence, without loss of generality we may assume that there exists $t_m-\alpha<\tau_n^{\prime}<\tau_n$ such that 
\begin{equation*}
\tau_n^{\prime}=\sup\{t:\;t_m-\alpha<t<\tau_n,\,\text{dist}(z_n^a(t),L_k)\leq\frac{\eta}{2}\}.
\end{equation*}
Since $z_n^a$ is equibounded in $L^{\infty}(I)$ we deduce from Remark \ref{affineclose} and (\ref{penalcontrol}) that
\begin{equation}\label{largepenal}
c_{\eta}:=\inf\{G(z_n(t)):\,t\in (\tau_n^{\prime},\tau_n)\}>0.
\end{equation}
Let $C_1>0$. By the assumptions on $\mu_n$, for $n$ large enough, applying the Cauchy-Schwarz inequality and (\ref{largepenal}) we have
\begin{align*}
C&\geq \frac{\mu_n}{\sqrt{2}\Ln\dn^{\frac{3}{2}}}\int_{\tau_n^{\prime}}^{\tau_n} G(z_n(s))\,\mathrm{d}s+\frac{\Ln}{\sqrt{2\dn}}\int_{\tau_n^{\prime}}^{\tau_n}|(z_n^a)^{\prime}(s)|^2\,\mathrm{d}s
\\
&\geq C_1\int_{\tau_n^{\prime}}^{\tau_n}\sqrt{G(z_n(s))} |(z_n^a)^{\prime}(s)|\,\mathrm{d}s\geq C_1\sqrt{c_{\eta}}|z_n^a(\tau_n)-z_n^a(\tau_n^{\prime})|\geq C_1\sqrt{c_{\eta}}\frac{\eta}{2}.
\end{align*}
This yields a contradiction since $C_1$ was arbitrary, so (\ref{backtothm}) holds. From this, arguing as in the proof of Theorem \ref{main3}, the lower bound follows.

The construction of a recovery sequence is the same as in Theorem \ref{main3} since we can assure that $G(z_n)=0$ up to minor details in the case when $\frac{1}{2}$ is not a lattice point. This can solved by modifying the function $f_{\e}$ defined in (\ref{optvelocity}) such that $f_{\e}(t)=0$ if $|t|\leq \e$ increasing the test energy only by $C\e$.
\end{proof}

Having described the different effects of the penalization term that prescribes the possible directions of the chiralty vector, there is only one case where we see a dependence on the distance between two chiral vectors, namely case (iii) in Theorem \ref{mainangle}. Since the nonlinear constraint $w=u\times u^{\prime}$ is non-trivial we are not able to solve the optimal profile problem explicitly. However we can qualitatively discuss an example where the transition energy is not constant.

\begin{example}\label{anglenonconstant}
Let $k=2$, let $Q_{2}=\{\pm q_{1},\pm q_{2}\}$, where $q_1= e_1$ and $q_2=(\cos(\alpha),\sin(\alpha),0)$ for $\alpha>0$. Consider $G_{|S^2}=\text{dist}(\cdot,Q_2)$. Since $G$ is $1$-Lipschitz, we can take $c_G(\eta)=\eta$ as modulus of continuity in Lemma \ref{connectpoints}. By the construction there, it follows that there exists a universal constant $C$, not depending on $G$, nor on $\alpha$, such that
\begin{equation*}
h_G(q_1,q_2)\leq C|q_1-q_2|.
\end{equation*}
Now, for $w\in H^1(\R,\R^3)$ such that $\lim_{t\to\pm\infty}w(t)=\pm q_1$ we take the continuous representative. Then, provided we have chosen $|\alpha|$ suitably small, the balls $B_{\frac{1}{8}}(\pm q_1)$ contain $\pm q_2$. Due to continuity, there exists an interval $(t_-, t_+)$ such that $w(t)\notin B_{\frac{1}{4}}(\pm q_1)$ and $\frac12\le |w(t)|$ for all $t\in (t_-, t_+)$ . If now $c$ is the strictly positive infimum of $\sqrt{\frac{G}{2}}$ on $\R^3\setminus \left(B_{\frac{1}{4}}(q_1)\cup B_{\frac{1}{4}}(-q_1) \cup B_{\frac{1}{2}}(0)\right)$, we have
\begin{equation*}
\int_{\R} (|w(t)|^2-1)^2+\frac{G(w(t))}{2}+|w^{\prime}(t)|^2\dt\geq c\int_{t_-}^{t_+}|w^{\prime}(t)|\dt\geq c|w(t_+)-w(t_-)|\geq\frac{c}{2}.
\end{equation*}
Thus, $h_G(q_1,-q_1)\geq \frac{c}{2}$. Since the constants $c$ and $C$ are independent of $\alpha$, for a suitable choice of $\alpha$ we will have 
\begin{equation*}
h_G(q_1,-q_1)>h_G(q_1,q_2).
\end{equation*}
Therefore the transition energy is in this case actually depending on the left and right limit of $z$ at a discontinuity point, differently than in the case of hard penalization.  

We also notice that with a similar argument we can show that $h_G(q_1,q_2)>0$, therefore it is not possible to have transitions with $0$ energy, differently than in the case of Theorem \ref{main1}.
\end{example}

\section{Main results for the helical XY-model}
%%%%%%%%%%%%%%%%%%%%%%%%%%%%%%%%%%%%%%%%%%%%%%%%%%%%%%%%%%%%%%%%%%%
Thanks to the results of the previous section, we can now study the asymptotic behavior of the renormalized energy $H_n$ defined in \eqref{2d-normalized}, scaled by $\Ln\dn^{\frac32}$, in the limit of {\it strong ferromagnetic interaction}. To be precise, we will assume that 
\begin{equation}\label{scaleregime}
\frac{\Ln}{\sqrt{\dn}}\to 0,\quad J_{2,n}\frac{\Ln}{\sqrt{\dn}}\geq c>0.
\end{equation}
As it is usual in the variational analysis of discrete systems we embed the energies in a common function space. To this end we identify every function $u\in \overline{{\mathcal U}}^2_{n}(\Omega)$ with its piecewise-constant interpolation belonging to the space
\begin{equation}\label{PC2}
C_{n,2}(\Omega,S^{2}):=\{u\in\overline{\mathcal{U}}^2_n(\Omega): u(x)=u(\Ln i)\,\,\hbox{ if } x \in \Ln(i+[0,1)^2),\,i \in \Z^2_n(\Omega)\}. 
\end{equation}
Thus we can extend the functional $H_n$ to a functional on defined on $L^{\infty}(\Omega,\R^3)$ by setting 
\begin{equation*}
H_n(u):=
\begin{cases}
H_n(u) &\mbox{if $u\in C_{n,2}(\Omega,S^2)$,}\\
+\infty &\mbox{otherwise.}
\end{cases}
\end{equation*}

The analysis of one-dimensional slices (see Remark \ref{noncompact}) has already shown that without further constraints on the spin variable we cannot expect $L^1$-compactness for sequences of bounded energy. This is why we add a penalization term, not changing the minimal energy of the system, to the normalized helical $XY$-model. Given a function $G$ as in Section \ref{anglesect} and $u\in C_{n,2}(\Omega,S^2)$, we namely set
\begin{equation*}
P^G_n(u)=\dn^2\sum_{i\in R_n(\Omega)}\Ln^2G(u^{i}\times u^{i+e_1}),
\end{equation*}
For $u\in C_{n,2}(\Omega,S^2)$ we define $z\in C_{n,2}(\Omega,\R^3)$ via
\begin{equation*}
z^{i}=\frac{u^{i}\times u^{i+e_1}}{\sqrt{2\dn}}
\end{equation*}
and write for short $z=T_2(u)$. Now we can define the energy $H^G_n:L^1(\Omega,\R^3)\rightarrow[0,+\infty]$ setting
\begin{equation}\label{2d-infimum}
H^G_n(z)=
\begin{cases}
\inf_{T_2(u)=z}H_n(u)+P^G_n(u) &\mbox{if $z=T_2(u)\text{ for some }u\in C_{n,2}(\Omega,S^2)$,}
\\
+\infty &\mbox{otherwise.}
\end{cases}
\end{equation}
Observe that we only deal with a two-dimensional analogue of the energy in Theorem \ref{mainangle} (iii). A hard penalization like in paragraph \ref{hard}, or a different scaling of the additional term $P^G_n$ like in Theorem \ref{mainangle} (iv) could be also considered, and the arguments we are going to use here would lead also in those cases to the analog of the results discussed in the one-dimensional case. We prefer anyway to focus on the choice of $P^G_n$ that gives in our opinion more significant results.

We now state and prove a compactness result for the energy $H^G_n(z)$.

\begin{proposition}\label{2d-compact}
Assume that (\ref{scaleregime}) holds and let $z_n\in L^1(\Omega,\R^3)$ be such that
\begin{equation*}
H^{G}_n(z_n)\leq C\Ln\dn^{\frac{3}{2}}.
\end{equation*}
Then (up to subsequences) $z_n$ converges strongly in $L^1$ to a function $z\in BV(\Omega,Q_k)$ depending only on $x$.
\end{proposition}

\begin{proof}
By choosing appropriate candidates for the infimum problem defining $H_n^G(z_n)$ there exists a sequence $u_n\in C_{n,2}(\Omega,S^2)$ such that $T_2(u_n)=z_n$ and $H^G_n(u_n)\leq (C+1)\Ln\dn^{\frac{3}{2}}$. 
Notice that for a.e.\ $y\in I$ the function $u_n(\cdot,y)$ is an element of $C_n(I; S^2)$ and $z_n(\cdot,y)=T_n(u_n(\cdot,y)$. Since $u_n$ is a piecewise constant function we have by the assumptions and the definitions \eqref{2d-infimum} and \eqref{penalenergy} of $H^G_n$ and $H_n^p$ that
\begin{equation}\label{sliceassumpt}
C+1\geq\int_{0}^1\frac{H_n^{p}(u_n(\cdot,y))}{\sqrt{2}\Ln\dn^{\frac{3}{2}}}\,\mathrm{d}y\,.
\end{equation}
Applying Fatou's lemma we deduce that, for almost every $y\in I$, we have
\begin{equation}\label{aefinite}
\liminf_n\frac{H_n^{p}(u_n(\cdot,y))}{\sqrt{2}\Ln\dn^{\frac{3}{2}}}<+\infty,
\end{equation}
so that by Lemma \ref{penalcompact}, for almost every $y\in I$  $z_{n}(\cdot,y)$ is compact in $L^1(I;\R^3)$. In particular, given a countable dense set $\mathcal{D}\subset I$, by a diagonal argument there exists a common subsequence $n_d$ such that $z_{n_d}(\cdot,y)$ is converging for all $y\in\mathcal{D}$. 

Without loss of generality we assume $1 \in \mathcal{D}$. We now show that the sequence $z_{n_d}$ is a Cauchy-sequence in $L^1(\Omega,\R^3)$. Let $\e>0$ and consider $0=y_0<y_1<\dots<y_N=1\in\mathcal{D}\cup\{0,1\}$ such that $\sup_{l}|y_{l+1}-y_l|\leq\frac{2}{N}$. Then we have
\begin{align*}
\int_{\Omega}&|z_n(x,y)-z_m(x,y)|\,\mathrm{d}(x,y)=\int_0^1\sum_{l=1}^{N}\int_{y_{l-1}}^{y_{l}}|z_n(x,y)-z_m(x,y)|\,\mathrm{d}y\,\mathrm{d}x
\\
\leq&\int_0^1\sum_{l=1}^{N}\int_{y_{l-1}}^{y_{l}}|z_n(x,y)-z_n(x,y_l)|+|z_n(x,y_l)-z_m(x,y_l)|+|z_m(x,y_l)-z_m(x,y)|\,\mathrm{d}y\,\mathrm{d}x
\\
\leq&\sum_{l=1}^N|y_{l}-y_{l-1}|\int_0^1|z_n(x,y_l)-z_m(x,y_l)|\,\mathrm{d}x
\\&\quad\quad+\int_0^1\sum_{l=1}^N\int_{y_{l-1}}^{y_l}|z_n(x,y)-z_n(x,y_l)|+|z_m(x,y)-z_m(x,y_l)|\,\mathrm{d}y\,\mathrm{d}x.
\end{align*}
We start by bounding the last term. To this end note that, for any $n\in\mathbb{N}$,
\begin{align}\label{averageclose}
\sum_{l=1}^N&\int_{y_{l-1}}^{y_l}|z_n(x,y)-z_n(x,y_l)|\,\mathrm{d}y\leq \sum_{l=1}^N\frac{2}{N}\sup_{y\in [y_{l-1},y_l]}|z_n(x,y)-z_n(x,y_l)|\nonumber
\\
&\leq \sum_{l=1}^N\frac{2}{N}\sum_{j\in\Z_n([y_{l-1},y_l])}|z_n(x,j+\Ln)-z_n(x,j)|\leq \frac{2}{N}\sum_{j\in\Z_n(I)}|z_n(x,j+\Ln)-z_n(x,j)|
\end{align}
From the very definition of $z_n$ we infer
\begin{align}\label{relationy}
\sqrt{2\dn}|z_n^{i+e_2}&-z_n^{i}|=|(u_n^{i+e_2}\times u_n^{(i+e_1)+e_2})-(u_n^{i}\times u_n^{i+e_1})|\nonumber
\\
&\leq |u_n^{i+e_2}\times(u_n^{(i+e_1)+e_2}-u_n^{i+e_1})|+|(u_n^{i+e_2}-u_n^{i})\times u_n^{i+e_1}|\nonumber
\\
&\leq |u_n^{(i+e_1)+e_2}-u_n^{i+e_1}|+|u_n^{i+e_2}-u_n^{i}|.
\end{align}
Combining (\ref{averageclose}), (\ref{relationy}) and integrating with respect to $x$ we deduce from the periodic boundary conditions that
\begin{equation*}
\int_0^1\sum_{l=1}^N\int_{y_{l-1}}^{y_l}|z_n(x,y)-z_n(x,y_l)|\,\mathrm{d}y\,\mathrm{d}x\leq \sqrt{2}\sum_{i\in R_n(\Omega)}\Ln^2\left|\frac{u_n^{i+e_2}-u_n^{i}}{N\sqrt{\dn}\Ln}\right|.
\end{equation*}
Applying Jensen's inequality we obtain
\begin{align*}
\left(\int_0^1\sum_{l=1}^N\int_{y_{l-1}}^{y_l}|z_n(x,y)-z_n(x,y_l)|\,\mathrm{d}y\,\mathrm{d}x\right)^2&\leq \frac{2}{N^2\Ln^2\dn}\sum_{i\in R_n(\Omega)}\Ln^2|u_n^{i+e_2}-u_n^{i}|^2
\\
\leq &\left(\frac{H_n^G(z_n)}{\sqrt{2}\Ln\dn^{\frac{3}{2}}}\right)\left(\frac{\sqrt{\dn}}{\Ln J_{2,n}}\right)\left(\frac{2}{N^2}\right)\leq \frac{C}{N^2}.
\end{align*}
For $N$ large enough, this yields
\begin{equation}\label{cauchy}
\int_{\Omega}|z_n(x,y)-z_m(x,y)|\,\mathrm{d}(x,y)\leq\sum_{l=1}^N|y_{l}-y_{l-1}|\int_0^1|z_n(x,y_l)-z_m(x,y_l)|\,\mathrm{d}x+\frac{\e}{2},
\end{equation}
from which we deduce the Cauchy property taking $m,n$ large enough for the finitely many Cauchy sequences $z_{n_d}(\cdot,y_l)$. 

Let now $z\in L^1(\Omega,\R^3)$ be the limit of $z_{n_d}$. By Lemma \ref{penalcompact}, $z(\cdot,y)\in BV(I,Q_k)$ for almost every $y\in I$, and therefore  $z\in L^1(\Omega,Q_k)$. 
We now claim that $z$ is independent of $y$. Observe that if this holds, we immediately get $z\in BV(\Omega,Q_k)$. 

We are therefore only left to prove the claim. To this aim, for $x\in I$ we denote by $z^a_{x,n}\in H^1(I,\R^3)$ the piecewise affine interpolation between the points $\{z_n(x,\Ln j)\}_{j=0}^{\left[1/\Ln\right]-1}$. Since $z_{n_d}$ converges to $z$ in $L^1(\Omega;\R^3)$, up to a subsequence, independent of $x$ and that we do not relabel, $z_{n_d}(x,\cdot)$ converges to $z(x, \cdot)$ in $L^1(I;\R^3)$ for a.e. $x \in I$. Then also the piecewise affine interpolations $z^a_{x,n_d}(\cdot)$ converge to $z(x,\cdot)$ in $L^1(I;\R^3)$for almost every $x\in I$. Furthermore, by definition
\begin{equation*}
\int_0^1\int_0^1|(z^a_{x,n_d})^{\prime}(y)|^2\,\mathrm{d}y\,\mathrm{d}x=\sum_{i\in R_{n_d}(\Omega)}\lambda_{n_d}^2\left|\frac{z_{n_d}^{i+e_2}-z_{n_d}^{i}}{\lambda_{n_d}}\right|^2. 
\end{equation*}
Using now \eqref{scaleregime} and \eqref{relationy}, and since $u_{n_d}\in C_{{n_d},2}(\Omega;S^2)$, we deduce
\begin{align*}
\int_0^1\int_0^1|(z^a_{x,n_d})^{\prime}(y)|^2\,\mathrm{d}y\,\mathrm{d}x\leq \frac{1}{\delta_{n_d}}\sum_{i\in R_{n_d}(\Omega)}\lambda_{n_d}^2\left|\frac{u_{n_d}^{i+e_2}-u_{n_d}^{i}}{\lambda_{n_d}^2}\right|^2
\leq\left(\frac{H_{n_d}^G(z_{n_d})}{\sqrt{2}\lambda_{n}\delta_{n}^{\frac{3}{2}}}\right) \left(\frac{\sqrt{2\delta_{n_d}}}{\lambda_{n_d} J_{2,n_d}}\right)\leq C \,.
\end{align*}
By Fatou's Lemma this implies
$$
\liminf_{n_d} \int_0^1|(z^a_{x,n_d})^{\prime}(y)|^2\,\mathrm{d}y<+\infty
$$
for almost every $x\in I$.  It follows that $z(x,\cdot)\in H^1(I,\R^3)$ and since it takes only finitely many values we have that $z(x,\cdot)$ is constant, which yields the claim.
\end{proof}

Concerning the $\Gamma$-limit of the rescaled and normalized energies, we obtain the following result.

\begin{theorem}\label{main2d}
Let $H_n^{G}:L^1(\Omega,\R^3)\rightarrow [0,+\infty]$ be defined as in (\ref{newfunctional}). Assume that $\frac{\Ln}{\sqrt{\dn}}\to 0$. Then the functionals $\frac{H_n^{G}}{\sqrt{2}\Ln\dn^{\frac{3}{2}}}$ $\Gamma$-converge with respect to the strong $L^1$-topology to the functional 
\begin{equation*}
H^{G}(z)=
\begin{cases}
\int_{S(z)}h_G(z_-,z_+)\,\mathrm{d}\mathcal{H}^1 &\mbox{if $z\in BV(\Omega,Q_k)$ does not depend on $y$,}\\
+\infty &\mbox{otherwise,}
\end{cases}
\end{equation*}
where $h_G$ is defined by \eqref{penaltransition} and $z_-$, $z_+$ are the one-sided limits of $z$ at a discontinuity point.
\end{theorem}

\begin{proof}
For the lower bound consider a sequence $z_n\in C_{n,2}(\Omega,\R^3)$ converging in $L^1(\Omega,\R^3)$ to some $z$ such that $\sup_n\frac{H_n^{G}}{\sqrt{2}\Ln\dn^{\frac{3}{2}}}\leq C$. By Proposition \ref{2d-compact} we get immediately that $z\in BV(\Omega,Q_k)$ and that $z$ does not depend on $y$. Let us denote by $z_0\in BV(I,Q_k)$ the slice of $z$ in the $x$-direction. For each $n$ let $u_n\in C_{n,2}(\Omega,M_k)$ be such that $T_2(u_n)=z_n$ and 
\begin{equation*}
\frac{H_n(u_n)+P^G_n(u_n)}{\sqrt{2}\Ln\dn^{\frac{3}{2}}}\leq \frac{H_n^G(z_n)}{\sqrt{2}\Ln\dn^{\frac{3}{2}}}+\frac{1}{n}.
\end{equation*}
From Theorem \ref{mainangle} we deduce
\begin{align*}
\liminf_n\frac{H_n(u_n)+P^G_n(u_n)}{\sqrt{2}\Ln\dn^{\frac{3}{2}}}&\geq \int_0^1\liminf_n\frac{H_n^{p}(u_n(\cdot,y))}{\sqrt{2}\Ln\dn^{\frac{3}{2}}}\,\mathrm{d}y
\\
&\geq\int_0^1\sum_{t\in S(z_0)}h_G(z_-,z_+)\,\mathrm{d}y=\int_{S(z)}h_G(z_-,z_+)\,\mathrm{d}\mathcal{H}^1,
\end{align*}
where we used that $z$ does not depend on $y$.

The upper bound is proved by taking a recovery sequence $\tilde u_n \in C_n(I, S^2)$ for the one dimensional energy $H^p_n$ defined by \eqref{penalenergy}. Setting $u^i_n= \tilde u^{i_1}_n$ for all $i=(i_1, i_2)\in R_n(\Omega)$ we obviously have $|u^{i+e_2}-u^i|=0$ for all $i$ and the result follows from Theorem \ref{mainangle} (iii).
\end{proof}

\bibliographystyle{plain}
\bibliography{bib-CS}
\end{document}